\def\llncs{0}
\def\fullpage{1}
\def\anonymous{0}
\def\authnote{1}
\def\draft{0} % 1 -> show keys, 0 -> hide keys
\def\notxfont{0}
\def\submission{0} %when submit 30page version to conference, it is 1. For full version, it is 0.

\ifnum\submission=1
\def\anonymous{1}
\def\llncs{1}
\def\fullpage{0}
\fi

\ifnum\llncs=1 	
    \documentclass[envcountsect,a4paper,runningheads]{llncs}
\else
	\documentclass[letterpaper,hmargin=1.05in,vmargin=1.05in,11pt]{article}
			\ifnum\fullpage=1
		\usepackage{fullpage}
		\fi
\fi

\usepackage[%
  colorlinks=true,
  citecolor=blue,
  pagebackref=true
]{hyperref}

\usepackage{amsmath, amsfonts, amssymb, mathtools,amscd}

\usepackage{amsthm}

\usepackage{lmodern}
\usepackage[T1]{fontenc}
\usepackage[utf8]{inputenc}

\usepackage{arydshln} % In order to use \hdashline
\usepackage{url}
\usepackage{ifthen}
\usepackage{bm}
\usepackage{multirow}
\usepackage[dvips]{graphicx}
\usepackage[usenames]{color}
\usepackage{xcolor,colortbl} % Leave out in case the usepackage cannot be found. Not critical
\usepackage{threeparttable}
\usepackage{comment}
\usepackage{paralist,verbatim}
\usepackage{cases}
\usepackage{booktabs}
\usepackage{breakcites}
\usepackage{braket}
\usepackage{cancel} 
\usepackage{ascmac} 
\usepackage{framed}
\usepackage{authblk}
\usepackage{pifont}
\usepackage{qcircuit}
\usepackage{tikz}
\usepackage{xspace} 
\usetikzlibrary{decorations.pathmorphing,decorations.shapes,decorations.markings}
\definecolor{darkblue}{rgb}{0,0,0.6}
\definecolor{darkgreen}{rgb}{0,0.5,0}
\definecolor{maroon}{rgb}{0.5,0.1,0.1}
\definecolor{dpurple}{rgb}{0.2,0,0.65}

\usepackage[capitalise,noabbrev]{cleveref}
\usepackage[absolute]{textpos}
\usepackage[final]{microtype}
\usepackage[absolute]{textpos}
\usepackage{everypage}
\DeclareMathAlphabet{\mathpzc}{OT1}{pzc}{m}{it}

\usepackage{algorithmic}
\usepackage{algorithm}
\usepackage{here}

\ifnum\draft=1
    \usepackage[color]{showkeys}
    \definecolor{refkey}{cmyk}{0,0,0,.25}
    \definecolor{labelkey}{cmyk}{0,0,0,.7}
    \usepackage{CJKutf8}
\else
\fi

\newtheoremstyle{thicktheorem}%
{\topsep}
{\topsep}
{\itshape}{}%
{\bfseries}%
{.}
{ }%
{\thmname{#1}\thmnumber{ #2}%
		\thmnote{ (#3)}%
}

\newtheoremstyle{remark}%name
{\topsep}
{\topsep}
	{}%body font
	{}%indent amount
	{}%theorem head font
	{.}%punctuation after theorem head
	{ }%space after theorem head
	{\textit{\thmname{#1}}\thmnumber{ #2}%theorem head specs
			\thmnote{ (#3)}%
	}

\ifnum\llncs=0
	\theoremstyle{thicktheorem}
	\newtheorem{theorem}{Theorem}[section]
	\newtheorem{lemma}[theorem]{Lemma}
	\newtheorem{corollary}[theorem]{Corollary}
	
	\newtheorem{definition}[theorem]{Definition}

	\theoremstyle{remark}
	\newtheorem{claim}[theorem]{Claim}
	\newtheorem{remark}[theorem]{Remark}

\else
\fi

\Crefname{MyClaim}{Claim}{Claims}
	\crefname{theorem}{Theorem}{Theorems}
	\crefname{assumption}{Assumption}{Assumptions}
	\crefname{construction}{Construction}{Constructions}
	\crefname{corollary}{Corollary}{Corollaries}
	\crefname{conjecture}{Conjecture}{Conjectures}
	\crefname{definition}{Definition}{Definitions}
	\crefname{exmaple}{Example}{Examples}
	\crefname{experiment}{Experiment}{Experiments}
	\crefname{counterexample}{Counterexample}{Counterexamples}
	\crefname{lemma}{Lemma}{Lemmata}
	\crefname{observation}{Observation}{Observations}
	\crefname{proposition}{Proposition}{Propositions}
	\crefname{remark}{Remark}{Remarks}
	\crefname{claim}{Claim}{Claims}
	\crefname{fact}{Fact}{Facts}
	\crefname{note}{Note}{Notes}

\ifnum\llncs=1
 \crefname{appendix}{App.}{Appendices}
 \crefname{section}{Sec.}{Sections}
\else
\fi

\ifnum\llncs=1
\renewcommand*{\backref}[1]{}
\else
	\renewcommand*{\backref}[1]{(Cited on page~#1.)}
	\ifnum\notxfont=1
	\else
		\usepackage{newtxtext}
	\fi
\fi
\ifnum\authnote=0  %%%%% Remove comments %%%%%
\newcommand{\mor}[1]{}
\newcommand{\minki}[1]{}
\newcommand{\takashi}[1]{}
\newcommand{\xagawa}[1]{}

\else
\newcommand{\mor}[1]{$\ll$\textsf{\color{red} Tomoyuki: { #1}}$\gg$}
\newcommand{\takashi}[1]{$\ll$\textsf{\color{orange} Takashi: { #1}}$\gg$}
\newcommand{\minki}[1]{$\ll$\textsf{\color{darkgreen} Minki: { #1}}$\gg$}

\newcommand{\xagawa}[1]{$\ll$\textsf{\color{magenta} Keita: { #1}}$\gg$}
\fi

\newcommand{\StateGen}{\mathsf{StateGen}}

\newcommand{\cA}{\mathcal{A}}
\newcommand{\cB}{\mathcal{B}}
\newcommand{\cC}{\mathcal{C}}
\newcommand{\cD}{\mathcal{D}}

\newcommand{\cG}{\mathcal{G}}

\newcommand{\cO}{\mathcal{O}}

\newcommand{\cS}{\mathcal{S}}

\def\makeuppercase#1{
\expandafter\newcommand\csname tl#1\endcsname{\widetilde{#1}}
}

\def\makelowercase#1{
\expandafter\newcommand\csname tl#1\endcsname{\widetilde{#1}}
}

\newcommand{\Z}{\mathbb{Z}}

\newcommand{\secp}{\lambda}

\newcommand{\ct}{\keys{ct}}

\newcommand*{\subpr}{\mathrm{pr}}
\newcommand*{\subHaarPR}{\mathrm{HaarPR}}
\newcommand*{\subHDDH}{\mathrm{HaarDDH}}
\newcommand*{\subDDH}{\mathrm{DDH}}
\newcommand*{\subNR}{\mathrm{NR}}
\newcommand*{\subHaar}{\mathrm{Haar}}

\newcommand*{\keys}[1]{\mathsf{#1}}

\newcommand*{\algo}[1]{\ensuremath{\mathsf{#1}}}

\newenvironment{boxfig}[2]{\begin{figure}[#1]\fbox{\begin{minipage}{0.97\linewidth}
                        \vspace{0.2em}
                        \makebox[0.025\linewidth]{}
                        \begin{minipage}{0.95\linewidth}
            {{
                        #2 }}
                        \end{minipage}
                        \vspace{0.2em}
                        \end{minipage}}}{\end{figure}}

\newcommand{\bit}{\{0,1\}}

\newcommand{\KeyGen}{\algo{KeyGen}}

\newcommand{\Enc}{\algo{Enc}}
\newcommand{\Dec}{\algo{Dec}}

\newcommand{\Ver}{\algo{Ver}}

\newcommand{\negl}{{\mathsf{negl}}}

\newcommand{\poly}{{\mathrm{poly}}}

\newcommand{\NR}{NR\xspace}

\makeatletter
\DeclareRobustCommand
  \myvdots{\vbox{\baselineskip4\p@ \lineskiplimit\z@
    \hbox{.}\hbox{.}\hbox{.}}}
\makeatother

\title{Quantum Group Actions}

\ifnum\anonymous=1
\ifnum\llncs=1
\author{\empty}\institute{\empty}
\else
\author{}
\fi
\else
\ifnum\llncs=1
\author{
Tomoyuki Morimae\inst{1} \and Keita Xagawa\inst{2} } 
\institute{
 Yukawa Institute for Theoretical Physics, Kyoto University, Kyoto, Japan \and TII
}
\else
\author[1]{Tomoyuki Morimae}
\author[2]{Keita Xagawa}
\affil[1]{{\small Yukawa Institute for Theoretical Physics, Kyoto University, Kyoto, Japan}\authorcr{\small tomoyuki.morimae@yukawa.kyoto-u.ac.jp} }
\affil[2]{{\small Technology Innovation Institute, Abu Dhabi, UAE}\authorcr{\small keita.xagawa@tii.ae}}
\fi %%%%% END OF LNCS branch
\fi

\date{}

\begin{document}

\maketitle

\begin{abstract}
In quantum cryptography, there could be a new world, Microcrypt, where
cryptography is possible but one-way functions (OWFs) do not exist.
Although many fundamental primitives and useful applications have been found in Microcrypt,
they lack ``OWFs-free'' concrete hardness assumptions on which they are based.
In classical cryptography, 
many hardness assumptions on concrete mathematical problems have been introduced,
such as the discrete logarithm (DL) problems or 
the decisional Diffie-Hellman (DDH) problems 
on concrete group structures related to finite fields or elliptic curves.
They are then abstracted to generic hardness assumptions such as the DL and DDH assumptions over group actions.
Finally, based on these generic assumptions, primitives and applications are constructed.
The goal of the present paper is to introduce several 
abstracted generic hardness assumptions in Microcrypt,
which could connect the concrete mathematical hardness assumptions with applications.
Our assumptions are based on a quantum analogue of group actions.
A group action is a tuple $(G,S,\star)$ of a group $G$, a set $S$, and an operation $\star:G\times S\to S$.
We introduce a quantum analogue of group actions, which we call quantum group actions (QGAs), 
where $G$ is a set of unitary operators, $S$ is a set of states, and
$\star$ is
the application of a unitary on a state. 
By endowing QGAs with some reasonable hardness assumptions,
we introduce a natural quantum analogue of the decisional Diffie-Hellman (DDH) assumption and 
pseudorandom group actions. 
Based on these assumptions, we construct classical-query pseudorandom function-like state generators (PRFSGs).
PRFSGs are a quantum analogue of pseudorandom functions (PRFs),
and have many applications such as IND-CPA SKE, EUF-CMA MAC, and private-key quantum money schemes.
Because classical group actions are instantiated with many concrete mathematical hardness assumptions, 
our QGAs could also have some concrete (even OWFs-free) instantiations.

\end{abstract}

\ifnum\submission=0
\newpage
\setcounter{tocdepth}{2}
\tableofcontents
\newpage
\fi

\section{Introduction} \label{sec:introduction}

\paragraph{Background.}
In classical cryptography, the existence of one-way functions (OWFs) is the minimum assumption~\cite{FOCS:ImpLub89}, 
because many primitives (such as pseudorandom generators (PRGs), pseudorandom functions (PRFs),
zero-knowledge, commitments, digital signatures, and secret-key encryptions (SKE)) are equivalent to OWFs in terms of existence, 
and almost all primitives (including public-key encryption (PKE) and multi-party computations) imply OWFs. 

On the other hand, recent active studies have demonstrated that in quantum cryptography, 
OWFs would not necessarily be the minimum assumption.
Many fundamental primitives have been introduced, such as pseudorandom unitaries (PRUs)~\cite{C:JiLiuSon18}, pseudorandom function-like state generators (PRFSGs)~\cite{C:AnaQiaYue22,TCC:AGQY22}, unpredictable state generators (UPSGs)~\cite{EPRINT:MorYamYam24},
pseudorandom state generators (PRSGs)~\cite{C:JiLiuSon18}, one-way state generators (OWSGs)~\cite{C:MorYam22},
EFI pairs~\cite{ITCS:BCQ23}, and one-way puzzles (OWPuzzs)~\cite{STOC:KhuTom24}. 
They seem to be weaker than OWFs~\cite{Kre21,STOC:KQST23,STOC:LomMaWri24},
but still imply many useful applications such as 
commitments~\cite{C:MorYam22,C:AnaQiaYue22,ITCS:BCQ23,AC:Yan22}, multi-party computations~\cite{C:MorYam22,C:AnaQiaYue22}, 
message authentication codes (MAC)~\cite{C:AnaQiaYue22,EPRINT:MorYamYam24},
secret-key encryptions (SKE)~\cite{C:AnaQiaYue22,EPRINT:MorYamYam24}, 
digital signatures~\cite{C:MorYam22}, 
private-key quantum money~\cite{C:JiLiuSon18}, 
etc.

In classical cryptography,
many hardness assumptions on concrete mathematical problems have been introduced,
such as the discrete logarithm (DL) problems or 
the decisional Diffie-Hellman (DDH) problems 
on concrete group structures related to finite fields or elliptic curves.
They are then abstracted to generic hardness assumptions such as the DL and DDH assumptions over group actions.
Finally, based on these generic assumptions, primitives and applications are constructed.

On the other hand, in quantum cryptography, the first step has not yet been studied.
Because 
PRUs can be constructed from OWFs~\cite{HuangMa},
and
PRUs imply PRFSGs, UPSGs, PRSGs, OWSGs, EFI pairs, and OWPuzzs, 
all of them can also be constructed from OWFs. (See \cref{fig:relations} for the relations.) 
However, no ``OWFs-free'' concrete mathematical hardness assumptions on which they are based are known.\footnote{See \cref{sec:related_works}.}

\subsection{Our Results}
The goal of the present paper is to introduce several abstracted generic hardness assumptions, which could
connect the concrete mathematical hardness assumptions with applications. 
As we will explain later, these new assumptions are a quantum analogue of cryptographic group actions~\cite{C:BraYun90,cryptoeprint:2006/291,TCC:JQSY19,AC:ADMP20}.
Because classical group actions have many concrete instantiations~\cite{PQCRYPTO:JaoDeFo11,AC:CLMPR18,FFTA:DAlDiS24},
our quantum versions of group actions could also have concrete (even OWFs-free) instantiations by considering natural quantum analogue
of classical hard problems.

Based on these quantum assumptions, we construct classical-query PRFSGs.
PRFSGs are a quantum analogue of PRFs. A PRFSG is a quantum polynomial-time (QPT) algorithm $\StateGen$ that takes a classical key $k$ and a bit string $x$ as input,
and outputs a quantum state $|\phi_k(x)\rangle$. The security roughly means that no QPT adversary can distinguish whether it is querying to $\StateGen(k,\cdot)$ with a random $k$ or
an oracle that outputs Haar random states, which we call the Haar oracle.\footnote{More precisely, the oracle works as follows. If it gets $x$ as input and $x$ was not queried before,
it samples a Haar random state $\psi_x$ and returns it. If $x$ was queried before, it returns the same state $\psi_x$ that was sampled before when $x$ was queried for the first time.}
PRFSGs imply almost all known primitives such as UPSGs, PRSGs, OWSGs, OWPuzzs, and EFI pairs.
PRFSGs also imply useful applications such as IND-CPA SKE, EUF-CMA MAC, 
private-key quantum money, commitments, multi-party computations, (bounded-poly-time-secure) digital signatures, etc.

Unfortunately, PRFSGs that we construct in this paper are secure only against classical queries.\footnote{IND-CPA SKE and EUF-CMA MAC
constructed from such PRFSGs are also secure against classical queries.}
It is an open problem whether PRFSGs secure against quantum queries 
or even PRUs can be constructed from quantum group actions.

\paragraph{Group actions.}
Our quantum assumptions are based on a ``quantization'' of group actions.
A group action $(\star,G,S)$ is a tuple of a group $G$, a set $S$, and
an operation $\star:G\times S\to S$ 
such that $g_1\star (g_2\star x)=(g_1g_2)\star s$ for any $g_1,g_2\in G$ and $s\in S$.
Cryptographic group actions~\cite{cryptoeprint:2006/291,TCC:JQSY19,AC:ADMP20,C:BraYun90} are group actions endowed with some hardness assumptions. 
For example, a one-way group action~\cite{C:BraYun90} is a group action such that given $s\gets S$\footnote{In this paper, $s\gets S$ means that
an element $s$ is sampled uniformly at random from the set $S$.} and $t\coloneqq g\star s$ 
with $g\gets G$,
it is hard to find a $g'$ such that $g'\star s=t$.
One-way group actions are abstractions of several well-studied cryptographic assumptions such as
the Discrete-Log assumptions~\cite{DifHel76}, %\mor{assumption? ikamo}
isogeny-based assumptions~\cite{PQCRYPTO:JaoDeFo11,AC:CLMPR18}, and 
code-based assumptions~\cite{FFTA:DAlDiS24}. 
They have several applications such as identifications, digital signatures, and commitments~\cite{C:BraYun90}. 

A pseudorandom group action~\cite{AC:ADMP20,TCC:JQSY19} is a group action such that $(s,g\star s)$ and $(s,u)$ are computationally indistinguishable,
where $s$ is a (fixed) element in $S$, $u\gets S$, and $g\gets G$. 
Pseudorandom group actions are abstractions of several well-studied cryptographic assumptions such as the Decisional Diffie-Hellman (DDH) assumptions~\cite{DifHel76} 
and isogeny-based assumptions~\cite{AC:CLMPR18}.\footnote{While we can treat some code-based assumptions as group actions, they are unlikely to be weakly pseudorandom and weakly unpredictable with large samples~\cite{FFTA:DAlDiS24,EPRINT:BCDSK24}.} 
They also have attractive applications such as key exchange, smooth projective hashing, dual-mode PKE, two-message statistically sender-private OT, and PRFs~\cite{C:BraYun90,AC:ADMP20,TCC:JQSY19,cryptoeprint:2006/291}.

%\mor{cite?}.\mor{how about one-way GA?} 
%and . 
%\mor{DDH and isogeny?}, 
% One-way group actions have several applications such as identifications, digital signatures, and commitments\cite{C:BraYun90}. 
%They have several applications such as key exchange, identifications, digital signatures, commitments, 
% Pseudorandom ones also have applications such as key exchange, smooth projective hashing, dual-mode PKE, two-message statistically sender-private OT, and PRFs~\cite{C:BraYun90,AC:ADMP20,TCC:JQSY19,cryptoeprint:2006/291}.
%symmetric KDM-secure encryption.
%\mor{All these assumptions are from PR GA? No application of One-way GA?}
%\xagawa{Bit commitment (stat.-hiding and comp.-binding) is constructed from one-way (certified) GA~\cite{C:BraYun90}}

\paragraph{Quantum group actions.}
In this paper, we introduce a quantum analogue of cryptographic group actions, which we call
\emph{quantum group actions (QGAs)}.
A QGA $(G,S,\star)$ is a tuple of a set $G$, a set $S$, and an operation $\star$.
$G$ is a set of efficiently-implementable unitary operators\footnote{Note that we do not require that $G$ is a group.}
and $S$ is a set of efficiently generable states. 
%\xagawa{If the word ends with "-ate", then we use "-able" instead of "-atable". E.g., calculate $\to$ calculable.}
The action $\star$ is just the application of a unitary in $G$ on a state in $S$.
Then the property $g_1(g_2|s\rangle)=(g_1g_2)|s\rangle$
is trivially satisfied for any $g_1,g_2\in G$ and $|s\rangle\in S$.

We % introduce several quantum assumptions by 
endow QGAs with several hardness assumptions.
In particular, we construct PRFSGs from these assumptions.
%To give the concrete hardness assumptions, we briefly review the DDH assumption and the Naor-Reingold PRF over group actions, 
% since we will construct PRFSGs by following the Naor-Reingold PRF. 

\paragraph{Naor-Reingold PRFs, DDH, and  (weak) pseudorandomness.}
To give an idea, we briefly review the classical construction of the Naor-Reingold (classical) PRFs~\cite{JACM:NaoRei04} based on some classical assumptions. 
The Naor-Reingold PRFs can be constructed from a group action as follows~\cite{JACM:NaoRei04,AC:BonKogWoo20,AC:ADMP20,AC:MorOnuTak20}. 
The key $k$ of the PRF $f_k$ is $k\coloneqq (g_0,g_1,...,g_\ell)$, where $g_i\gets G$ for $i=0,1,...,\ell$.
For an input $x=(x_1,...,x_\ell)\in\bit^\ell$,
$f_k(x)$ is defined as 
\begin{align}
 f_k(x)\coloneqq (g_\ell^{x_\ell} \cdot \dots \cdot g_1^{x_1} g_0) \star s_0, 
\end{align}
where $s_0$ is a fixed element in $S$.
Roughly speaking, its security is shown by the computational indistinguishability \footnote{Here $\approx_c$ means that the two distributions are computationally indistinguishable.}
\begin{align}
 \{(g_i \star s_0, (\tilde{g}g_i) \star s_0) : \tilde{g}, g_i \gets G\}_{i \in [Q]}
 \approx_c 
 \{(g_i \star s_0, h_i \star s_0) : g_i, h_i \gets G \}_{i \in [Q]}, \label{cind:GA:NR}
\end{align} 
which, for clearness, we call the Naor-Reingold (NR) assumption. 
Here, $Q$ is a polynomial of the security parameter. 
Applying the NR assumption repeatedly, Naor and Reingold showed that $f_k(x)$ is computationally indistinguishable from $f'_k(x) \coloneqq g_x \star s_0$, where $g_x \gets G$ for each $x$~\cite{JACM:NaoRei04}. 
% between $\{(g_i \star s_0, (\tilde{g}g_i) \star s_0)\}_{i \in [Q]}$ and $\{(g_i \star s_0, h_i \star s_0)\}_{i \in [Q]}$, where $\tilde{g},g_i,h_i \gets G$ and $Q$ is polynomial of the security parameter. 
%\mor{When the group action is regular, $g_x\star s_0$ with $g_x\gets G$ is equivalent to the uniform sample $s\gets S$, and therefore pseudorandomness is satisfied.}\mor{$\leftarrow$ is OK?}\xagawa{OK}
%\xagawa{Or, we can consider $f'$ as a random function. It depends on the definition of a random function.}

% We consider a quantum analogue of the Decisional Diffie-Hellman (DDH) assumption~\cite{DiffieHellman}.
% The DDH assumption is a versatile tool for constructing classical public-key primitives. 
Naor and Reingold~\cite{JACM:NaoRei04} showed that the NR assumption is derived from the DDH  assumption. The DDH assumption says that
%it is hard to distinguish
\begin{align}
(s_0, \tilde{g} \star s_0, g \star s_0, (\tilde{g}g) \star s_0)
\approx_c
(s_0, \tilde{g} \star s_0, g \star s_0, h \star s_0), \label{cind:GA:DDH}
\end{align}
%$(s_0, \tilde{g} \star s_0, g \star s_0, \tilde{g}g \star s_0)$ and $(s_0, \tilde{g} \star s_0, g \star s_0, h \star s_0)$, 
where $s_0$ is a fixed element in $S$ and $\tilde{g},g,h \gets G$. 
If $G$ is a \emph{commutative ring} with $(\cdot,+)$ and $S$ has a binary operation $\circ$ such that $(g \star s_0) \circ (g' \star s_0) = (g + g') \star s_0$, 
%\mor{what is $\circ$ and what is +?}
then the DDH assumption tightly implies the NR assumption, 
because we can re-randomizing the samples~\cite{JACM:NaoRei04,CCS:BonMonRag10}.\footnote{\cite{CCS:LewWat09,C:EHKRV13,C:AbdBenPas15} treated some non-commutative cases related to the Matirx DDH assumptions.}
%  and removing $\tilde{g} \star s_0$
%\xagawa{Added Boneh, Kogan, and Woo~\cite{AC:BonKogWoo20}. They assume that \emph{transitive} and \emph{faithful} group action with commutative group $G$.}
Boneh, Kogan, and Woo~\cite{AC:BonKogWoo20} considered the case that $G$ is a \emph{commutative group} and 
showed that the DDH assumption implies the NR assumption via a hybrid argument. 
%\xagawa{Alamati~et~al.~\cite{AC:ADMP20} only considered Abelian group $G$...}
%To treat a non-commutative group $G$ and non-algebraic set $S$, 
Alamati, De~Feo, Montgomery, and Patranabis~\cite{AC:ADMP20} took a different approach; they defined weak pseudorandomness,\footnote{Correctly speaking, they defined it as the assumption that $\pi_{\tilde{g}}: s \mapsto \tilde{g} \star s$ is a weak pseudorandom \emph{permutation}.}
 which is the computational indistinguishability 
\begin{align}
 \{(s_i, \tilde{g} \star s_i) : \tilde{g} \gets G, s_i \gets S\}_{i \in [Q]}
 \approx_c 
 \{(s_i, s_i') : s_i,s_i' \gets S \}_{i \in [Q]}. \label{cind:GA:wPR1}
\end{align} 
%which says that it is hard to distinguish $\{(s_i, g \star s_i)\}_{i \in [Q]}$ from $\{(s_i,s'_i)\}_{i \in [Q]}$ with $s_i,s'_i \gets S$ and $g \gets G$. 
If the group action is regular,\footnote{A group action is \emph{regular} if it is 
\emph{transitive}, that is, for every $s_1,s_2 \in S$, there exists $g \in G$ satisfying $s_2 = g \star s_1$, 
 and \emph{free}, that is, for each $g \in G$, $g$ is the identity element if and only if there exists $s \in S$ satisfying $s = g \star s$~\cite{AC:ADMP20}.} % ~\cite{AC:ADMP20}, 
%\mor{what is regular?}, 
then the distribution $s_i \gets S$ is equivalent to the distribution of $g_i \star s_0$ with $g_i \gets G$. % (This means that this group action is trivially pseudorandom.)
% Thus, the weak pseudorandomness is equivalent to the computational indistinguishability 
% \begin{align}
%  \{(s_i, g \star s_i) : g \gets G, s_i \gets S\}_{i \in [Q]}
%  \approx_c 
%  \{(s_i, h_i \star s_i) : h_i \gets G, s_i \gets S \}_{i \in [Q]}. \label{cind:GA:wPR2}
% \end{align} 
%two distributions can be considered as $\{(s_i, g \star s_i)\}_{i \in [Q]}$ from $\{(s_i,h_i \star s_i)\}_{i \in [Q]}$ with $s_i \gets S$ and $g, h_i \gets G$, respectively. 
Thus, by replacing $s_i$ and $s'_i$ with $g_i \star s_0$ and $h_i \star s_0$, where $g_i \gets G$ and $h_i \gets G$, respectively, the weak pseudorandomness is tightly equivalent to the NR assumption. 

We note that while Alamati~et~al.~focused only on the case that $G$ is commutative, their approach can be extended to non-commutative groups $G$. We also note that we will not need some properties of $G$ in the proof in Boneh~et~al.~\cite{AC:BonKogWoo20} when we employ pseudorandom group actions. %the PR assumption
For details, see \cref{sec:ClassicNRPRF}. 
% ~ommitted the proof~\cite[Lemma 4.20]{AC:ADMP20} and we cannot know how they proved their lemma.) 

%\paragraph{Pseudorandomness, Haar-Pseudorandomness, and quantum DDH.}
\paragraph{Construction of PRFSGs.}

\begin{figure}
\centering
%\footnotesize
\ifnum\fullpage=0
\scriptsize
\begin{tikzpicture}[scale=0.64,
base/.style={thick,->},]
\else 
\footnotesize
\begin{tikzpicture}[scale=0.7,
base/.style={thick,->},]
\fi 
\node (Base)  at ( 0,0) {\shortstack[c]{Haar-PR (eq.\ref{cind:QGA:Haar-PR}) \\ + Haar-DDH (eq.\ref{cind:QGA:Haar-DDH})}};
\node (comDH) at ( 0,-2) {\shortstack[c]{$G$'s commutativity \\ + DDH}};
\node (wPR)   at ( 5,0) {Weak PR (eq.\ref{cind:QGA:tmp})};
\node (NR)    at (10,0) {NR (eq.\ref{cind:QGA:NR})}; 
\node (PRFSGs) at (15,0) {PRFSGs (eq.\ref{NR-PRFSG})};
\node (PR)    at ( 5,2) {PR (eq.\ref{cind:QGA:PR})};
\draw[base] (Base) to node {} (wPR);
\draw[base] (wPR)  to node {} (NR);
\draw[base] (NR)   to node {} (PRFSGs);
% \draw[base] (wPR)  to node [below,midway] {\cref{lem:PRetc->NR}} (NR);
% \draw[base] (NR)   to node [below,midway] {\cref{thm:PRNR->PRFSG}} (PRFSGs);
\draw[thick](PR)   to[out=0,in=180] node {} ([xshift=-2cm] NR) ;
\draw[thick](PR)   -- (10,2) to[out=0,in=180] node {} ([xshift=-2cm] PRFSGs) ;
\draw[base](comDH) -- (9,-2) to[out=0,in=270] (NR.south) ;
\ifnum\fullpage=0
\end{tikzpicture}
\else 
\end{tikzpicture}
\fi
\caption{Diagram for our construction.}
% the relations among assumptions on QGAs and PRFSGs.}
\label{fig:PRFSG}
\end{figure}
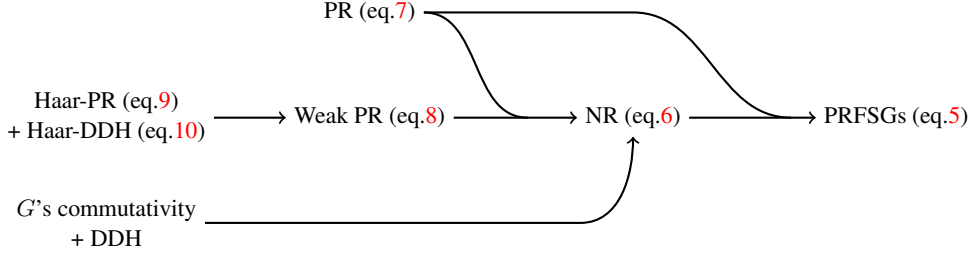

Based on these classical constructions of the Naor-Reingold PRFs, we try to construct PRFSGs.
Jumping ahead, our construction is summarized in~\cref{fig:PRFSG}. 
Let $(G,S)$ be a QGA.\footnote{We omit $\star$, because this is trivial.}
This means that $G$ is a set of efficiently implementable unitary operators and $S$ is a set of efficiently generatable states.
We will construct a PRFSG, $\StateGen(k,x)\to|\phi_k(x)\rangle$, as follows: 
The key $k$ of the PRFSG is $k \coloneqq (g_0,g_1,...,g_\ell, |s_0\rangle)$, 
where $g_i\gets G$ for $i=0,1,...,\ell$ and $|s_0\rangle$ is a (fixed) element in $S$.
For an input $x=(x_1,...,x_\ell)\in\bit^\ell$,
the output of PRFSG $|\phi_k(x)\rangle$ is defined as\footnote{We note that Ananth, Gulati, and Lin~\cite{EPRINT:AnaGulLin24} gave a similar construction of selectively-secure PRFSGs 
in the common Haar state model, which is inspired by GGM~\cite{JACM:GolGolMic86}.}
\begin{align}
 |\phi_k(x)\rangle \coloneqq (g_\ell^{x_\ell} \cdot \dots \cdot g_1^{x_1} g_0) |s_0\rangle.  \label{NR-PRFSG}
\end{align}
%$(Z^{\bigoplus_{i=1,\dots,m} k_i^{x_i}} |s_0\rangle$~\cite{TCC:AnaGulLin24}.}
%
%\mor{I thought their construction is not GGM but Naor-Reingold. No?}
%\xagawa{They cite GGM instead of NR and their proof is similar to that of GGM, while their construction is very similar to NR.
%$(Z^{\bigoplus_{i=1,\dots,m} k_i^{x_i}}) |s_0\rangle = g_i^{x_i} |s_0\rangle$ with $g_i = Z^{k_i}$. 
%By the way, NR can be considered as a generalization of GGM as discussed in Boneh, Montgomery, and 
%Raghunatahn~\cite{CCS:BonMonRag10}.}

The question is which hardness assumptions should
we endow the QGA with so that $\StateGen$ satisfies the security of PRFSGs.
In quantum group actions, we cannot expect that $G$ has algebraic structures, and the simple analogue of the DDH assumption or/and weak pseudorandomness would not imply the quantum analogue of the NR assumption that roughly states the computational indistinguishability\footnote{Actually, our security game is such that the adversary receives many copies of the state.
Hence, the assumption should be read as the computational indistinguishability
$
\{(g_i|s_0\rangle, \tilde{g} g_i |s_0\rangle)^{\otimes t} : \tilde{g}, g_i \gets G \}_{i \in [Q]} 
\approx_c 
\{(g_i |s_0\rangle,h_i|s_0\rangle)^{\otimes t} : g_i,h_i \gets G \}_{i \in [Q]}
$
for any polynomial $t$.
However, in this introduction, we ignore the number of copies for ease of notation and use the word ``roughly''.}
\begin{align}
\{(g_i|s_0\rangle, \tilde{g} g_i |s_0\rangle) : \tilde{g}, g_i \gets G \}_{i \in [Q]} \approx_c \{(g_i |s_0\rangle,h_i|s_0\rangle) : g_i,h_i \gets G \}_{i \in [Q]}, \label{cind:QGA:NR}
\end{align}
where $|s_0\rangle$ is a (fixed) element in $S$. 
Thus, we need to put forth simple, plausible assumptions over quantum group actions that imply the quantum analogue of the NR assumption. 
%after applying the NR assumption repeatedly, we could show that 
%$|\phi_k(x)\rangle$ is indistinguishable from $|\phi'(x)\rangle \coloneqq g_x |s_0\rangle$ where $g_x \gets G$ for each $x$ and we need to show the oracle returns $|\phi'(x)\rangle$ is indistinguishable from the Haar oracle. 

In the quantum case, moreover, \cref{cind:QGA:NR} is not enough to construct PRFSGs unlike the classical case.
In the classical construction of NR PRFs, by applying the classical NR assumption repeatedly, we can show that
$f_k(x)$ is indistinguishable from $g_x\star s_0$ with $g_x\gets G$ for each $x$.
In the classical case, because of the regularity, $g_x\star s_0$ with $g_x\gets G$ is equivalent to sampling $s\gets S$. 
However, in the quantum case,
we do not have regularity in general, and we cannot expect that $g_x|s_0\rangle$ with $g_x\gets G$ is uniformly at random in some efficiently samplable set $S'$.\footnote{$S'$ might differ from $S$.} 
Thus, we will require the additional assumption that $g_x|s_0\rangle$ with $g_x\gets G$ is indistinguishable from Haar random states. 
We call this assumption \emph{pseudorandomness (PR)}, which roughly says the computational indistinguishability
\begin{align}
 (|s_0\rangle, h |s_0\rangle) \approx_c (|s_0\rangle,|s'\rangle), \label{cind:QGA:PR}
\end{align}
%between $(|s_0\rangle, g |s_0\rangle)$ and $(|s_0\rangle,|s'\rangle)$, 
where $|s_0\rangle$ is a (fixed) element in $S$, $h \gets G$, and $|s'\rangle \gets \mu$. (Here, $|s'\rangle\gets\mu$ means that a state $|s'\rangle$ is sampled uniformly at random with the Haar measure.)\footnote{Again, here, the computational indistinguishability is that for many copies of states, but for simplicity we omit it.}
By combining the quantum analogue of the NR assumption (\cref{cind:QGA:NR}) and this PR assumption, we get PRFSGs.

Then, the question is how can we get the quantum analogue of the NR assumption?
In the classical case, we get it from the weak pseudorandomness, \cref{cind:GA:wPR1},~\cite{AC:ADMP20}.
We can introduce a quantum analogue of it, 
which is the computational indistinguishability 
\begin{align}
\{(|s_i\rangle, \tilde{g} |s_i\rangle) : \tilde{g} \gets G, |s_i\rangle \gets \mu \}_{i \in [Q]} 
\approx_c \{(|s_i\rangle,|s'_i\rangle) : |s_i\rangle,|s_i'\rangle \gets \mu \}_{i \in [Q]}. \label{cind:QGA:tmp}
\end{align}
In the classical case, the weak pseudorandomness is equivalent to the classical NR assumption,
but in the quantum case, again because of the fact that we do not have regularity in general, \cref{cind:QGA:tmp}
will not imply the quantum analogue of the NR assumption, \cref{cind:QGA:NR}.
However, combining this with the PR assumption (\cref{cind:QGA:PR}), we will recover \cref{cind:QGA:NR}. 

% Looking the two indistinguishabilities, \cref{cind:GA:wPR1} and \cref{cind:GA:wPR2}, there are three ensembles: 
% \begin{itemize}
%     \item[a)]
% $\{(s_i, g \star s_i) : g \gets G, s_i \gets S\}_{i \in [Q]}$, 
% \item[b)]
% $\{(s_i, h_i \star s_i) : h_i \gets G, s_i \gets S \}_{i \in [Q]}$, 
% \item[c)]
% $\{(s_i, s_i') : s_i,s_i' \gets S \}_{i \in [Q]}$. 
% \end{itemize}
% \xagawa{We want to show that $ $. Applying PR assumption, $ $}

Therefore the goal is to realize the quantum analogue of weak pseudorandomness, \cref{cind:QGA:tmp}.
To achieve it, we put forth two new assumptions, which we believe plausible and reasonable: 
The one is %We call the quantum analogue of the computational indistinguishability of b) and c) 
\emph{Haar-pseudorandomness (Haar-PR)}, which roughly states the computational indistinguishability 
\begin{align}
(|s\rangle, h |s\rangle) \approx_c (|s\rangle,|s'\rangle), \label{cind:QGA:Haar-PR}
\end{align}
%$(|s\rangle, g |s\rangle)$ is computationally indistinguishable from $(|s\rangle,|s'\rangle)$ with 
where $|s\rangle, |s'\rangle \gets \mu$ and $h \gets G$.\footnote{Note that we give unbounded-polynomial copies of the sample to the adversary. If the number of copies is constant, then there exists a statistical construction~\cite[Section~4]{EPRINT:AnaGulLin24}.} 
Interestingly,  PR and Haar-PR are not equivalent, because, unlike the classical case with regularity, 
$h |s_0\rangle$ in the LHS of PR (\cref{cind:QGA:PR}) may not be distributed according to the Haar measure. 
%and  \mor{kokokara}$h |s\rangle$ in the LHS of Haar-PR (\cref{cind:QGA:Haar-PR}) may not be independent of $|s\rangle$.
% uniform sampling (i.e., Haar sampling) of elements (i.e., states) cannot be done efficiently.\footnote{Note that we give unbounded-polynomial copies of the sample to the adversary.}

The other is a quantum analogue of the DDH assumption with multiple samples and with respect to Haar random states. % computational indistinguishability of a) and b), 
We call it \emph{Haar-DDH}, which roughly states the computational indistinguishability 
\begin{align}
 \{(|s_i\rangle, g |s_i\rangle) : |s_i\rangle \gets \mu, g \gets G \}_{i \in [Q]} \approx_c \{(|s_i\rangle,h_i|s_i\rangle) : |s_i\rangle \gets \mu, h_i \gets G\}_{i \in [Q]}. \label{cind:QGA:Haar-DDH}
\end{align}
%$\{(|s_i\rangle, g |s_i\rangle)\}_{i \in [Q]}$ from $\{(|s_i\rangle,h_i|s_i\rangle)\}_{i \in [Q]}$, 
% where $s_i$ are chosen according to the Haar measure $\mu$ and $g, h_i \gets G$. 
%Since states are chosen randomly, it seems reasonable to assume this computational hardness. 
%\mor{This sentence is not clear to me.}

%In the quantum case, this quantum DDH assumption does not directly imply the NR assumption. 
%To bridge them, we introduce two additional reasonable assumptions. 

The combination of Haar-PR and Haar-DDH assumptions implies the quantum analogue of weak pseudorandomness, \cref{cind:QGA:tmp}. 
Combining it with the PR assumption, we can show a quantum analogue of the NR assumption (\cref{cind:QGA:NR}). 
We can show the security of our NR-style PRFSGs from the quantum analogue of the NR assumption and the PR assumption. 
%Fortunately, the combination of PR, Haar-PR, and quantum DDH assumptions implies a quantum analogue of the NR assumption. 
%, that is, the computational indistinguishability between $\{(g_i |s_0\rangle, (\tilde{g}g_i) |s_0\rangle)\}_{i \in [Q]}$ and $\{(g_i |s_0\rangle, h_i |s_0\rangle)\}_{i \in [Q]}$, where $\tilde{g},g_i,h_i \gets G$ and $Q$ is a polynomial of the security parameter. 
% Thus, by employing those three assumptions, we can show the security of our PRFSGs.

In general, PRFSGs are defined against quantum-query adversaries~\cite{C:AnaQiaYue22,TCC:AGQY22}. This means that the security holds against any QPT adversary
that can query $x$ in superposition.
Unfortunately, our proof only works for \emph{classical-query} cases, and there are several barriers 
to the construction of PRFSGs secure against quantum queries. For details, see~\cref{sec:PRFSG:quantum-query}. %\mor{(For details, see ??)}
It is an open problem to construct PRFSGs secure against quantum queries or even PRUs from
QGAs (or other OWFs-free assumptions).

In the classical case, PRFs can be constructed from pseudorandom group actions~\cite{TCC:JQSY19}.
On the other hand, we do not know how to construct PRFSGs from 
PR or Haar-PR QGAs. One reason is that the construction of PRFs in \cite{TCC:JQSY19} is the GGM one~\cite{FOCS:GolGolMic84}, and
we do not know how to use the GGM technique in the quantum setting. For example, we do not know how to hash quantum states.
Moreover, in the classical case, we can construct PRFs from PRGs~\cite{FOCS:GolGolMic84},
but it is an open problem whether we can construct PRFSGs from PRSGs.\footnote{PRFSGs with $O(\log)$ input length can be constructed from PRSGs~\cite{C:AnaQiaYue22}, but it is open for PRFSGs with $\poly$ input length.}
%PRFSGs are quantum analogue of PRFs introduced in \cite{C:AnaQiaYue22}. PRFSGs 
%have many applications such as IND-CPA SKE~\cite{C:AnaQiaYue22}, EUF-CMA MAC~\cite{C:AnaQiaYue22}, 
%PRSGs, and OWSGs~\cite{C:MorYam22}.
%PRFSGs could also exist even if (quantumly-secure) one-way functions do not exist~\cite{Kre21}.

% DDH is...\mor{please write}
% As its quantum analogue, we define quantum DDH as follows...\mor{please write}
% In classical case, we can construct PRF from DDH as follows. \mor{please write}. 
% However, in quantum case, a similar proof does not work because...\mor{please write}
% We therefore introduce the following two additional assumption...\mor{please write}

\paragraph{PRSGs from PR QGAs.}
As we have explained,
PR QGAs is the computational indistinguishability
$
 (|s_0\rangle, h |s_0\rangle)^{\otimes t} \approx_c (|s_0\rangle,|s'\rangle)^{\otimes t} 
$
for any polynomial $t$, 
where $|s_0\rangle$ is a (fixed) element in $S$, $h \gets G$, and $|s'\rangle \gets \mu$. 
% A natural quantum analogue of the pseudorandomness is the following one:
% $|s\rangle^{\otimes t}\otimes (g|s\rangle)^{\otimes t}$
% and
% $|s\rangle^{\otimes t}\otimes |s'\rangle^{\otimes t}$
% are computationally indistinguishable for any polynomial $t$, where $|s\rangle\gets S$, $g\gets G$, and $|s'\rangle\gets\mu$.
% We call a QGA that satisfies such pseudorandomness \emph{pseudorandom QGA (PR QGA)}.
% Another natural quantum analogue of the pseudorandomness
% is to replace $|s\rangle\gets S$ with $|s\rangle\gets\mu$ in the definition of PR QGAs.
% That is,
% $|s\rangle^{\otimes t}\otimes (g|s\rangle)^{\otimes t}$
% and
% $|s\rangle^{\otimes t}\otimes |s'\rangle^{\otimes t}$
% are computationally indistinguishable for any polynomial $t$, where $|s\rangle\gets \mu$, $g\gets G$, and $|s'\rangle\gets\mu$.
% We call a QGA that satisfies such pseudorandomness \emph{Haar pseudorandom QGA (Haar-PR QGA)}.
% Interestingly, these two definitions are not equivalent, because, unlike the classical case, 
% uniform sampling (i.e., Haar sampling) of elements (i.e., states) cannot be done efficiently.\footnote{Note that $t$ is unbounded polynomial.}
As an additional result, we observe that PRSGs can be constructed from PR QGAs.
\begin{lemma}
PR QGAs imply PRSGs. 
\end{lemma}
% \begin{theorem}
% PR QGAs exist if and only if pseudorandom state generators (PRSGs) exist.    
% \end{theorem}

\if0
PRSGs are a quantum analogue of PRGs introduced in \cite{C:JiLiuSon18}.
PRSGs could exist even if OWFs do not exist, but have many applications such as private-key quantum 
money~\cite{C:JiLiuSon18} and OWSGs~\cite{C:MorYam22}.
PRSGs are constructed from OWFs~\cite{C:JiLiuSon18}, but it is an open problem to capture PRSGs with some cryptographic assumptions
that will not imply OWFs.
The above result therefore shows that PRSGs can be characterised with a natural quantum analogue of pseudorandom group actions that do not seem to imply OWFs.
\fi

\paragraph{OWSGs from one-way QGAs.}
It is also natural to define a quantum analogue of one-way group actions.
In the security game of classical one-way group actions, the adversary receives classical bit strings $s$ and $g\star s$.
In our one-way QGAs, 
the adversary receives $|s\rangle^{\otimes t}$ and $(g|s\rangle)^{\otimes t}$ for any polynomial $t$.
We show the following.
\begin{lemma}
One-way QGAs imply pure one-way state generators (OWSGs). 
\end{lemma}

\paragraph{Candidates of QGAs.}
Finally, we briefly argue about some candidates for QGAs. We expect that
QGAs based on random quantum circuits and random IQP circuits 
are PR, Haar-PR, and DDH QGAs.

\paragraph{Open Problems.}
\cref{fig:relations} is a summary of the new and known relations between cryptographic primitives and QGAs, in which we separate primitives with classical-query and quantum-query securities. 
As is shown in the figure, our results could
open a new avenue to connect quantum cryptographic applications with concrete OWFs-free hardness assumptions.

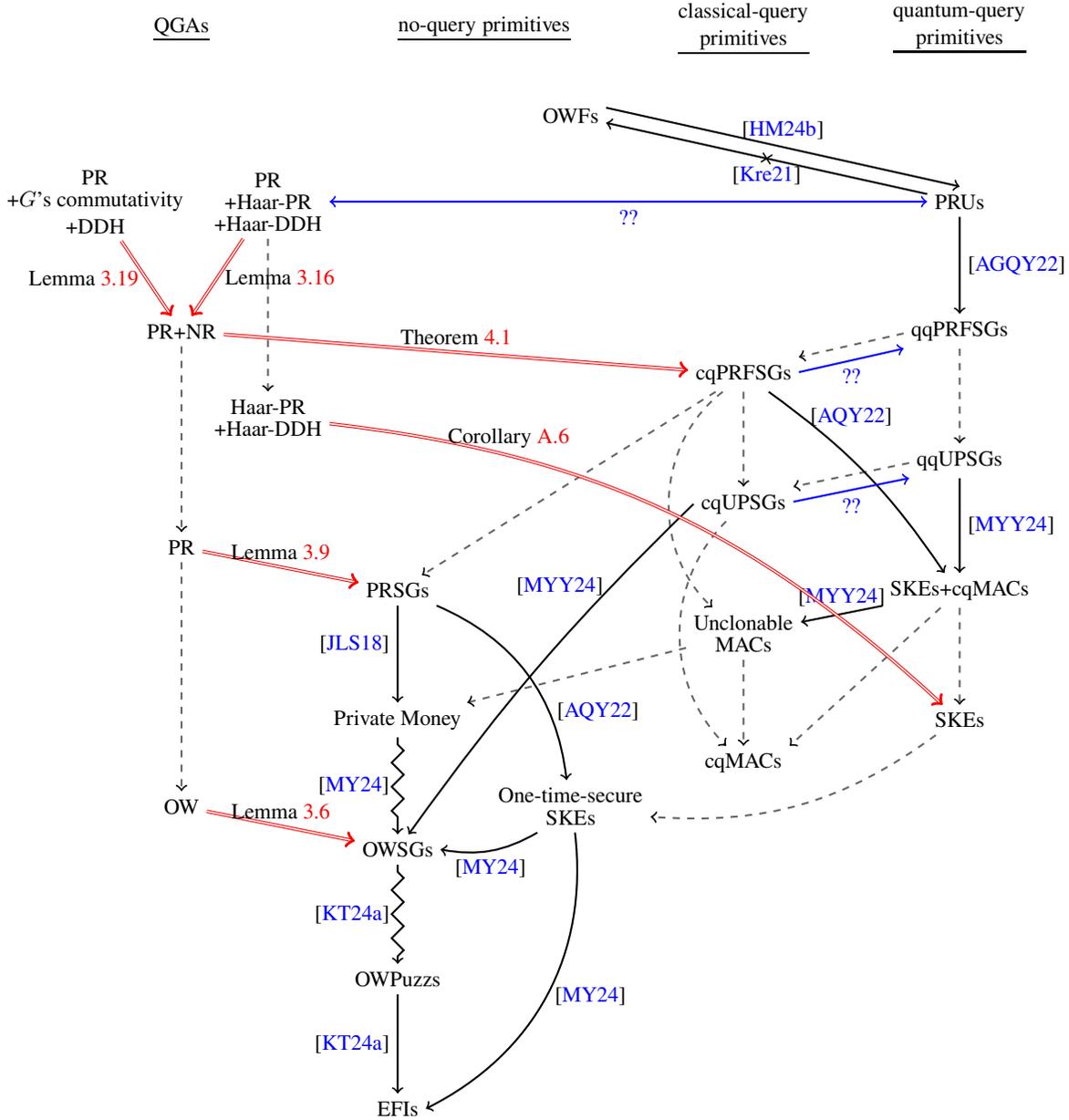
\begin{figure}
\centering
%\footnotesize
\ifnum\llncs=1
\scriptsize
\begin{tikzpicture}[scale=0.55,%
triv/.style={thick,draw=black!60,->,dashed},
prev/.style={thick,->},
pure/.style={thick,->,decorate,decoration={zigzag,pre length=3pt,post length=3pt}},
sepa/.style={thick,->,postaction={decorate,decoration={markings,mark=at position .5 with {\draw[semithick,-] (-2pt,-2pt) -- (2pt,2pt); \draw[semithick,-] (2pt,-2pt) -- (-2pt,2pt);}}}},
open/.style={thick,blue,->},
ours/.style={thin,draw=red,->,double},
]
\else 
\footnotesize
\begin{tikzpicture}[scale=0.63,%
triv/.style={thick,draw=black!60,->,dashed},
prev/.style={thick,->},
pure/.style={thick,->,decorate,decoration={zigzag,pre length=3pt,post length=3pt}},
sepa/.style={thick,->,postaction={decorate,decoration={markings,mark=at position .5 with {\draw[semithick,-] (-2pt,-2pt) -- (2pt,2pt); \draw[semithick,-] (2pt,-2pt) -- (-2pt,2pt);}}}},
open/.style={thick,blue,->},
ours/.style={thin,draw=red,->,double},
]
\fi 
\node at (  9,17) {\underline{\shortstack[c]{quantum-query \\ primitives}}};
\node at (  4,17) {\underline{\shortstack[c]{classical-query \\ primitives}}};
\node at ( -2,17) {\underline{no-query primitives}};
\node at ( -9,17) {\underline{QGAs}};
% \node at (-14,17) {\underline{NR QGAs}};
%
\node (OWF)     at (  0,15) {OWFs} ;
\node (PRU)     at (  9,13) {PRUs} ;
\node (qqPRFSG) at (  9,10) {qqPRFSGs} ;
\node (cqPRFSG) at (  4, 9) {cqPRFSGs} ;
\node (qqUPSG)  at (  9, 7) {qqUPSGs} ;
\node (cqUPSG)  at (  4, 6) {cqUPSGs} ;
\node (SKE)     at (  9, 1) {SKEs} ;
\node (QPOTP)   at (  0,-1) {\shortstack[c]{One-time-secure \\ SKEs}} ;
\node (SKE+MAC) at (  9, 4) {SKEs+cqMACs} ;
\node (ucMAC)   at (  4, 3) {\shortstack[c]{Unclonable \\ MACs}} ;
\node (MAC)     at (  4, 0) {cqMACs} ;
\node (PRSG)    at ( -4, 4) {PRSGs} ;
\node (PrivMon) at ( -4, 1) {Private Money} ; 
\node (OWSG)    at ( -4,-2) {OWSGs} ; 
\node (OWPuzz)  at ( -4,-5) {OWPuzzs} ;
\node (EFI)     at ( -4,-8) {EFIs} ; 
\node (PRHDHQGA)at ( -7,13) {\shortstack[c]{PR \\ +Haar-PR \\ +Haar-DDH}};
\node (PRcDHQGA)at (-11,13) {\shortstack[c]{PR \\ +$G$'s commutativity \\ +DDH}};
\node (PRNRQGA) at ( -9,10) {PR+NR} ;
\node (HDHQGA)  at ( -7, 8) {\shortstack[c]{Haar-PR \\ +Haar-DDH}} ;
\node (PRQGA)   at ( -9, 5) {PR} ;
\node (OWQGA)   at ( -9,-1) {OW} ;
\draw[triv] ([yshift=10pt] qqPRFSG.south west) to node [above,midway] {} ([yshift=0pt] cqPRFSG.north east);
\draw[triv] ([yshift=10pt]  qqUPSG.south west) to node [above,midway] {} ([yshift=0pt]  cqUPSG.north east);
\draw[triv] (qqPRFSG) to node [above,midway] {} (qqUPSG);
\draw[triv] (cqPRFSG) to node [above,midway] {} (cqUPSG);
\draw[triv] (cqPRFSG) to [bend right=50] node [above,midway] {} (ucMAC);
\draw[triv] (cqPRFSG) to node [above,midway] {} (PRSG); 
\draw[triv] (ucMAC)   to node [right,midway] {} (PrivMon);
\draw[triv] (SKE+MAC) to node [right,midway] {} (SKE);
\draw[triv] (SKE+MAC) to node [right,midway] {} (MAC.north east);
\draw[triv] (SKE)     to [bend left=20] node [right,midway] {} (QPOTP);
\draw[triv] (ucMAC)   to node [right,midway] {} (MAC);
\draw[triv] (cqUPSG)  to [bend right=45] node [left,midway] {} (MAC);
\draw[prev] ([yshift=5pt] OWF.east) to node [above,midway] {\cite{HuangMa}} (PRU.north); 
\draw[sepa] ([yshift=-5pt] PRU.north west) to node [below,midway] {\cite{Kre21}} ([yshift=-5pt] OWF.east);
%\draw[sepa] (PRU.west)  to node [above,midway] {\cite{Kre21}} (OWF.east);
% \draw[prev] (OWF.south)     to node [left,midway] {\cite{TCC:AGQY22}} (qqPRFSG); 
\draw[prev] (PRU)     to node [right,midway] {\cite{TCC:AGQY22}} (qqPRFSG); 
\draw[prev] (cqPRFSG) to [bend left=10] node [right,pos=.15] {\cite{C:AnaQiaYue22}} ([xshift=-10pt] SKE+MAC.north); 
\draw[prev] (SKE+MAC) to node [above,midway] {\cite{EPRINT:MorYamYam24}} (ucMAC); 
\draw[prev] (qqUPSG)  to node [right,midway] {\cite{EPRINT:MorYamYam24}} (SKE+MAC); 
\draw[prev] (QPOTP) to [bend left=20] node [below,midway] {\cite{TQC:MorYam24}} (OWSG.east); 
\draw[prev] (QPOTP) to [bend left=35] node [right,midway] {\cite{TQC:MorYam24}} (EFI.east);
\draw[pure] (PrivMon)  to node [left ,midway] {\cite{TQC:MorYam24}}  (OWSG); 
\draw[pure] (OWSG)     to node [left ,midway] {\cite{STOC:KhuTom24}}  (OWPuzz); 
\draw[prev] (OWPuzz)   to node [left ,midway] {\cite{STOC:KhuTom24}}  (EFI); 
\draw[prev] (PRSG) to node [left,pos=.4] {\cite{C:JiLiuSon18}}  (PrivMon);
\draw[prev] (PRSG) to [bend left] node [right,pos=.7] {\cite{C:AnaQiaYue22}} (QPOTP); 
\draw[prev] (cqUPSG.west)  to [bend right=5] node [left,pos=.25] {\cite{EPRINT:MorYamYam24}} (OWSG); 
%\draw[prev] (cqUPSG)  to [bend right] node [left,pos=.25] {\cite{EPRINT:MorYamYam24}} (OWSG); 
%\draw[thin,->] (cqPRFSG) to [bend left] node [right,midway] {\cite{TCC:AGQY22}} (oTCCSKE); 
%
\draw[triv] (PRNRQGA) to node [right,midway] {}  (PRQGA); 
\draw[triv] (PRQGA) to node [right,midway] {} (OWQGA);
\draw[triv] (PRHDHQGA) to node [right,midway] {} (HDHQGA);
\draw[ours] (PRHDHQGA) to node [right,midway] {\cref{lem:PRetc->NR}}   (PRNRQGA); 
\draw[ours] (PRcDHQGA) to node [left,midway] {\cref{lem:com+DDH->NR}}   (PRNRQGA); 
\draw[ours] (PRQGA)   to node [above,midway] {\cref{thm:PR->PRSG}}    (PRSG);
\draw[ours] (HDHQGA)   to [bend left=18] node [above,pos=.25] {\cref{cor:HaarPR+Haar-DDH->SKE}} (SKE);
\draw[ours] (PRNRQGA) to node [above,midway] {\cref{thm:PRNR->PRFSG}} (cqPRFSG);
\draw[ours] (OWQGA)   to node [above,midway] {\cref{thm:OW->OWSG}}    (OWSG);
\draw[open] ([yshift=-10pt] cqPRFSG.north east) to node [below,midway] {??} ([yshift=0pt] qqPRFSG.south west) ; 
\draw[open] ([yshift=-10pt]  cqUPSG.north east) to node [below,midway] {??} ([yshift=0pt]  qqUPSG.south west) ;
\draw[open,<->] (PRU) to node [below,midway] {??} (PRHDHQGA);
\ifnum\fullpage=0
\end{tikzpicture}
\else 
\end{tikzpicture}
\fi
\caption{
Relations among primitives and QGAs. 
%``comm.'' denotes the commutativity of arbitrary two elements in $G$ of QGAs. \mor{where comm is used?}
``cq'' and ``qq'' denote classical-query and quantum-query, respectively. % ``qm'' means that quantum-message. 
An arrow from primitive A to primitive B (except for those with crosses) represents that A implies B. 
An arrow with a cross represents that there exists a black-box separation between A and B.
A gray dashed arrow represents that A trivially implies B. %primitive A trivially implies primitive B. 
A zigzag arrow represents that A with pure outputs implies B. %primitive A with pure outputs implies primitive B.
A red double arrow represents that the relation is shown in this paper. 
A blue arrow from primitive A to primitive B represents that we do not know the implication from A to B.
%*: \cite{C:AnaQiaYue22} shows that cqPRFSG implies IND-CPA SKE. Combined with the relation that cqUPSG implies EUF-CMA MAC, cqPRFSG implies IND-CPA SKE + EUF-CMA MAC. 
% \ctie{C:AnaQiaYue22} shows that cqPRFSG implies both IND-CPA SKE and EUF-CMA MAC. 
cqPRFSGs with the input space $\bit^\ell$ for $\ell = O(\log(\secp))$ can be constructed from PRSGs~\cite{C:AnaQiaYue22}. 
%\xagawa{Q: Is there an arrow from cqMAC to OWSG? Maybe we can do it.}
}
\label{fig:relations}
\end{figure}

We leave some interesting open problems: 
\begin{enumerate}
\item Do PR, Haar-PR, and Haar-DDH assumptions over QGAs imply \emph{quantum-query} PRFSGs? Or, can we show the separation between quantum-query/classical-query PRFSGs? 
\item Can we construct PR, Haar-PR, and Haar-DDH QGAs from PRUs? 
% (It is easy to construct PR and Haar-PR QGAs from PRUs.) 
\item
Can we construct PRUs from PR, Haar-PR, and Haar-DDH QGAs, or from other ``genuinely quantum'' assumptions?
\end{enumerate}

\subsection{Related Works}
\label{sec:related_works}
As we have explained in Introduction, the important open problem is to base ``Microcrypt'' primitives
on ``OWFs-free'' concrete mathematical hardness assumptions. 
Recently, the following three papers that tackle the problem have been uploaded on arXiv
during the preparation of this manuscript.

Khurana and Tomer~\cite{cryptoeprint:2024/1490} constructed OWPuzzs from some hardness assumptions that imply
sampling-based quantum advantage~\cite{NatPhys:BFNV19,STOC:AarArk11,TD04,BreJozShe10,BreMonShe16,FKMNTT18}
(plus a mild complexity assumption, $\mathbf{P}^{\# \mathbf{P}}\not\subseteq (io)\mathbf{BQP}/\mathbf{qpoly}$).

Hiroka and Morimae~\cite{HM24_meta} and 
Cavalar, Goldin, Gray, and Hall [private communication]
constructed OWPuzzs from quantum-average-hardness of
GapK problem. GapK problem is a promise problem to decide whether a given bit string $x$ has a small Kolmogorov complexity or not.
Its quantum-average-hardness means that the instance $x$ is sampled from a quantum-polynomial-time samplable distribution, and
no quantum-polynomial-time algorithm can solve the problem.

Their assumptions are more concrete and already studied in other contexts than cryptography, namely, quantum advantage and (classical) meta-complexity.
On the other hand, the present paper construct PRFSGs (and therefore UPSGs, PRSGs, OWSGs, private-key quantum money schemes, IND-CPA SKE, EUF-CMA MAC, OWPuzzs, and EFI pairs).
It is an interesting open problem whether our QGAs assumptions can be instantiated with some
concrete assumptions related to quantum advantage or meta-complexity.

\section{Preliminaries} \label{sec:preliminaries}
\subsection{Basic Notations}
We use the standard notations of quantum information and cryptography.
For a finite set $X$, $x\gets X$ means that an element $x$ is sampled from $X$ uniformly at random.
We write $\mu_m$ to denote the Haar measure over $m$-qubits space. We often drop the subscription $m$. 
For an algorithm $\cA$, $y\gets\cA(x)$ means that $\cA$ is run on input $x$ and output $y$ is obtained.
For a non-negative integer $Q$, $[Q]$ is the set $\{1,2,...,Q\}$.
QPT stands for quantum polynomial time.
$\secp$ is the security parameter.
$\negl$ is a negligible function.
For two distributions $D$ and $D'$, we sometimes use $D \approx_c D'$ to denote $D$ and $D'$ are computationally indistinguishable with respect to a quantum adversary. 

\subsection{Quantum Cryptographic Primitives}
We review quantum cryptographic primitives in the literature. 
\begin{definition}[Pseudorandom State Generators (PRSGs)~\cite{C:JiLiuSon18}]
A pseudorandom state generator (PRSG) is a tuple $(\KeyGen,\StateGen)$ of algorithms such that
\begin{itemize}
    \item 
    $\KeyGen(1^\secp)\to k:$
    It is a QPT algorithm that, on input $1^\secp$, outputs a key $k$.
    \item 
    $\StateGen(k)\to |\phi_k\rangle:$
It is a QPT algorithm that, on input $k$,
outputs a quantum state $|\phi_k\rangle$. 
\end{itemize}
We require that for any QPT adversary $\cA$ and any polynomial $t$,
\begin{align}
\left|\Pr_{k\gets \KeyGen(1^\secp)}[1\gets\cA(1^\secp,|\phi_k\rangle^{\otimes t})]    
-\Pr_{|\psi\rangle\gets\mu}[1\gets\cA(1^\secp,|\psi\rangle^{\otimes t})]   \right|\le\negl(\secp), 
\end{align}
where $\mu$ is a Haar measure.
\end{definition}

\begin{definition}[One-Way State Generators (OWSGs)~\cite{C:MorYam22,TQC:MorYam24}]
A one-way state generator (OWSG) is a tuple $(\KeyGen,\StateGen,\Ver)$ of algorithms such that
\begin{itemize}
    \item 
    $\KeyGen(1^\secp)\to k:$
    It is a QPT algorithm that, on input $1^\secp$, outputs a classical bit string $k\in\bit^{\kappa(\secp)}$,
    where $\kappa$ is a polynomial.
    \item 
    $\StateGen(k)\to|\phi_k\rangle:$
    It is a QPT algorithm that, on input $k$, outputs a quantum state $|\phi_k\rangle$.
    \item 
    $\Ver(k',|\phi_k\rangle)\to\top/\bot:$
    It is a QPT algorithm that, on input $k'$ and $|\phi_k\rangle$, outputs $\top/\bot$.
\end{itemize}
We require the following correctness and one-wayness.
\paragraph{Correctness.}
\begin{align}
\Pr[\top\gets\Ver(k,|\phi_k\rangle):k\gets\KeyGen(1^\secp),|\phi_k\rangle\gets\StateGen(k)]\ge 1-\negl(\secp).    
\end{align}

\paragraph{One-wayness.}
For any QPT adversary $\cA$ and any polynomial $t$,
\ifnum\fullpage=1
\begin{align}
\Pr[\top\gets\Ver(k',|\phi_k\rangle):k\gets\KeyGen(1^\secp),|\phi_k\rangle\gets\StateGen(k),k'\gets\cA(1^\secp,|\phi_k\rangle^{\otimes t})]
\le\negl(\secp).    
\end{align}
\else
\begin{multline}
\Pr[\top\gets\Ver(k',|\phi_k\rangle):k\gets\KeyGen(1^\secp),|\phi_k\rangle\gets\StateGen(k),k'\gets\cA(1^\secp,|\phi_k\rangle^{\otimes t})] \\ 
\le\negl(\secp).    
\end{multline}
\fi

\if0
$\Pr[\top\gets\cC]\le\negl(\secp)$ is satisfied in the following security game.
\begin{enumerate}
    \item 
    The challenger $\cC$ runs $k\gets\KeyGen(1^\secp)$.
    \item 
    $\cC$ runs $|\phi_k\rangle\gets \StateGen(k)$ $t$ times.
    \item
    $\cC$ sends $|\phi_k\rangle^{\otimes t}$ to $\cA$.
    \item 
    $\cA$ returns $k'$.
    \item 
    $\cC$ projects $|\phi_k\rangle$ onto $|\phi_{k'}\rangle$. If the projection is successful,
    $\cC$ outputs $\top$. Otherwise, it outputs $\bot$.
\end{enumerate}
\fi
\end{definition}

\begin{remark}
If all $|\phi_k\rangle$ are pure and $\Pr[\top\gets\Ver(k,|\phi_k\rangle)]\ge1-\negl(\secp)$ is satisfied for all $k$, 
we can replace $\Ver$ with the following canonical verification algorithm:
Project $|\phi_k\rangle$ onto $|\phi_{k'}\rangle$. If the projection is successful, output $\top$.
    Otherwise, output $\bot$.
\end{remark}

\begin{definition}[Weak OWSGs~\cite{TQC:MorYam24}]
The definition of weak OWSGs is the same as that of OWSGs except that the one-wayness is replaced with the
following weak one-wayness:
there exists a polynomial $p$ such that for any QPT $\cA$ and polynomial $t$
\ifnum\fullpage=1
\begin{align}
\Pr[\top\gets\Ver(k',|\phi_k\rangle):k\gets\KeyGen(1^\secp),|\phi_k\rangle\gets\StateGen(k),k'\gets\cA(1^\secp,|\phi_k\rangle^{\otimes t})]\le 1-\frac{1}{p(\secp)}.
\end{align}
\else
\begin{multline}
\Pr[\top\gets\Ver(k',|\phi_k\rangle):k\gets\KeyGen(1^\secp),|\phi_k\rangle\gets\StateGen(k),k'\gets\cA(1^\secp,|\phi_k\rangle^{\otimes t})] \\ \le 1-\frac{1}{p(\secp)}.
\end{multline}
\fi
\end{definition}
\begin{remark}
It is shown in Theorem 3.7 of \cite{TQC:MorYam24} that
OWSGs exist if and only if weak OWSGs exist.
\end{remark}

% In order to define pseudorandom function-like state generators (PRFSGs) and its variants, 
%  we first define function-like state generators (FSGs). 
%  \mor{Because we do not explain unpredictable and unclonalbe one, we just directly define PRFSGs.}
% \begin{definition}[Function-Like State Generators (FSGs)]
% A function-like state generator (FSG) is a tuple $(\KeyGen,\StateGen)$    
% of algorithms such that
% \begin{itemize}
%     \item 
%     $\KeyGen(1^\secp)\to k:$
%     It is a QPT algorithm that, on input $1^\secp$, outputs $k\in\bit^{\kappa(\secp)}$.
%     \item 
%     $\StateGen(k,x)\to |\phi_k(x)\rangle:$
%     It is a QPT algorithm that, on input $k$ and $x\in\bit^\ell$, outputs a quantum state $|\phi_k(x)\rangle$.
% \end{itemize}
% \end{definition}

\begin{definition}[Pseudorandom Function-Like State Generators (PRFSGs)~\cite{C:AnaQiaYue22}] 
\label{def:PRFSGs}
A pseudorandom function-like state generator (PRFSG) is a tuple $(\KeyGen,\StateGen)$    
of algorithms such that
\begin{itemize}
    \item 
    $\KeyGen(1^\secp)\to k:$
    It is a QPT algorithm that, on input $1^\secp$, outputs $k\in\bit^{\kappa(\secp)}$,
    where $\kappa$ is a polynomial.
    \item 
    $\StateGen(k,x)\to |\phi_k(x)\rangle:$
    It is a QPT algorithm that, on input $k$ and $x\in\bit^\ell$, outputs a quantum state $|\phi_k(x)\rangle$,
    where $\ell$ is a polynomial.
\end{itemize}
We require the following security:
% We require that for any QPT adversary $\cA$,
For any QPT adversary $\cA$, 
\begin{align}
\left|\Pr_{k\gets\KeyGen(1^\secp)}[1\gets\cA^{\StateGen(k,\cdot)}(1^\secp)]    
-\Pr[1\gets\cA^{\cO_{\subHaar}}(1^\secp)] \right|\le\negl(\secp).   
% \left|\Pr_{k\gets\bit^{\kappa(\secp)}}[1\gets\cA^{\StateGen(k,\cdot)}(1^\secp)]    
% -\Pr[1\gets\cA^{\cO_{Haar}}(1^\secp)] \right|\le\negl(\secp).   
\end{align}
The oracle $\cO_{\subHaar}$ is the following oracle:
\begin{enumerate}
    \item  
    When $x$ is queried and it is not queried before, sample $|\psi_x\rangle\gets\mu$ and return $|\psi_x\rangle$.
    \item 
    When $x$ is queried and it was queried before, return $|\psi_x\rangle$.
\end{enumerate}
\end{definition}

\begin{definition}[Classical-Query PRFSGs]
A PRFSG is called a classical-query PRFSG if
it is secure against only adversaries that query the oracle classically.
\end{definition}

\begin{remark} \label{remark:PRFSGs:Query}
In \cite{C:AnaQiaYue22,TCC:AGQY22}, general PRFSGs where adversaries can quantumly query the oracle
are defined and constructed from PRUs or OWFs.
In this paper, however, we mainly focus on classical-query ones.
\end{remark}

\subsection{Design and Haar Measure}
We will use the following lemmas to show our results.
\begin{lemma}[Lemma 20 and Lemma 21 of \cite{Kre21}]
\label{lem:design-Haar}
For each $n,t\in \mathbb{N}$ and $\epsilon>0$, 
there exists 
a $\poly(n, t, \log\frac{1}{\epsilon})$-time quantum algorithm $\cS$ 
that outputs an $n$-qubit state such that 
for any quantum algorithm $\cA$
\begin{align}
(1-\epsilon)\Pr_{|\psi\rangle\gets\mu}
[1\gets\cA(|\psi\rangle^{\otimes t})] 
\le 
\Pr_{|\psi\rangle\gets \cS}
[1\gets\cA(|\psi\rangle^{\otimes t})]
\le (1+\epsilon)\Pr_{|\psi\rangle\gets\mu}
[1\gets\cA(|\psi\rangle^{\otimes t})].
\end{align}
\end{lemma}

% \begin{lemma}[...]
% \[
% \mathbb{E}_{|s\rangle,|s'\rangle \gets \mu}[\langle s | s'\rangle]
% = 0
% \]
% \end{lemma}
\begin{lemma}\label{lem:Haar-orthogonal}
\begin{align}
\mathbb{E}_{|\psi\rangle,|\phi\rangle\gets \mu_n}|\langle \psi|\phi\rangle|^2\le \frac{1}{2^n}.    
\end{align}    
\end{lemma}
\begin{proof}
It is known that $\mathbb{E}_{|\psi\rangle\gets\mu_n}|\psi\rangle\langle \psi|=\frac{I^{\otimes n}}{2^n}$,    
where $I\coloneqq|0\rangle\langle0|+|1\rangle\langle1|$ is the two-dimentional identity operator.
Therefore, 
\begin{align}
\mathbb{E}_{|\psi\rangle,|\phi\rangle\gets \mu_n}|\langle \psi|\phi\rangle|^2
=\mathbb{E}_{|\phi\rangle\gets\mu_n}\langle \phi|[
\mathbb{E}_{|\psi\rangle\gets\mu_n}|\psi\rangle\langle\psi|]
|\phi\rangle
= \frac{1}{2^n}.    
\end{align}
\end{proof}

\section{Quantum Group Actions and Hardness Assumptions} \label{sec:QGA}
In this section, we define quantum group actions (QGAs) and
endow them with several hardness assumptions including one-wayness and variants of pseudorandomness.

\subsection{Quantum Group Actions}
We first define quantum group actions (QGAs).
\begin{definition}[Quantum Group Actions (QGAs)]
A quantum group action (QGA) is a pair $(G,S)$ of algorithms such that
\begin{itemize}
\item 
$G(1^\secp)\to [g]:$ 
It is a QPT algorithm that takes $1^\secp$ as input, and outputs an efficient classical description $[g]$ of a unitary
operator $g$. 
\item
$S(1^\secp)\to [|s\rangle]:$
It is a QPT algorithm that takes $1^\secp$ as input, and outputs an efficient classical description $[|s\rangle]$ of a quantum state $|s\rangle$.
\end{itemize}
\end{definition}
\begin{remark}
Note that we do not require that
the set $\{g:[g]\gets G(1^\secp)\}$ is a group of unitary operators. 
However, we call $(G,S)$ a quantum group action, because it is a quantum analogue of a group action.
\end{remark}
\begin{remark}
Note that $G$ and $S$ are not deterministic. This means that
each execution of $G$ or ($S$) can output different $[g]$ (or $[|s\rangle]$).
\end{remark}
\begin{remark}
An efficient classical description $[g]$ of $g$ means, for example, a classical description of a $\poly(\secp)$-size quantum circuit that implements $g$.
An efficient classical description $[|s\rangle]$ of $|s\rangle$ means, for example, a classical description of a $\poly(\secp)$-size quantum circuit that generates $|s\rangle$.
For simplicity, we often write $[g]$ and $[|s\rangle]$ just as $g$ and $|s\rangle$, respectively, if there is no confusion.
\end{remark}

\if0
In order to do analogues discussions with classical group actions, we define
the sets
\begin{align}
S_\secp\coloneqq\{|\psi_j\rangle:j\in\bit^{\sigma(\secp)},|\psi_j\rangle\gets \cS(1^\secp,j)\}
\end{align}
and
\begin{align}
G_\secp\coloneqq\{U_k:k\in\bit^{\sigma(\secp)},U_k|\psi\rangle\gets \cG(1^\secp,k,|\psi\rangle)\}.
\end{align}
Moreover, we often omit the security parameter $\secp$ and just write $S$ and $G$.
\fi

\subsection{One-Way QGAs} \label{sec:QGA:OW}
We next define a quantum analogue of one-way group actions. 

\begin{definition}[One-Way QGAs (OW QGAs)]
A QGA $(G,S)$ is called a one-way QGA (OW QGA) if
for any QPT adversary $\cA$ and any polynomial $t$,
$\Pr[\top\gets\cC]\le\negl(\secp)$ is satisfied in the following security game.
\begin{enumerate}
    \item 
    The challenger $\cC$ runs $[|s\rangle]\gets S(1^\secp)$ and $[g]\gets G(1^\secp)$.
    \item
    $\cC$ sends $|s\rangle^{\otimes t}$ and $(g|s\rangle)^{\otimes t}$ to $\cA$.
    \item 
    $\cA$ returns an efficient classical description $[g']$ of a unitary $g'$.\footnote{$[g']$ could be outside of the support of $G$.} %\mor{$[g']$ could be outside of the support of $G$.}}
    \item 
    $\cC$ projects $g|s\rangle$ onto $g'|s\rangle$.
    If the projection is successful,
    $\cC$ outputs $\top$. Otherwise, it outputs $\bot$.
\end{enumerate}
\end{definition}

\begin{lemma} \label{thm:OW->OWSG}
If OW QGAs exist then OWSGs exist.    
\end{lemma}
\begin{proof}
\if0
We first show the if part. Let $(\KeyGen,\StateGen)$ be a OWSG.
From it, we construct a OW QGA as follows:
\begin{itemize}
    \item 
    $\cG(1^\secp,k,|\psi\rangle)$: Run $|\phi_k\rangle\gets\StateGen(k)$, and output $|\phi_k\rangle$.
    \item 
    $\cS(1^\secp,j)$: Always output $|0...0\rangle$.
\end{itemize}
Assume that it is not a OW QGA. Then there exist a QPT adversary $\cA$ and polynomials $t,p$ such that
$\Pr[\top\gets\cC]\ge\frac{1}{p(\secp)}$ for infinitely many $\secp$ in the following security game.
\begin{enumerate}
\item 
$\cC$ chooses $k\gets\bit^{\kappa(\secp)}$. \xagawa{Should we choose $k$ according to $\KeyGen$ of OWSG?}
    \item
    $\cC$ sends $|0...0\rangle^{\otimes t}\otimes|\phi_k\rangle^{\otimes t}$ to $\cA$.
    \item 
    $\cA$ returns $k'$.
    \item 
    $\cC$ projects $|\phi_k\rangle$ onto $|\phi_{k'}\rangle$. If the projection is successful,
    $\cC$ outputs $\top$. Otherwise, it outputs $\bot$.
\end{enumerate}
It is easy to see that from this $\cA$ we can construct a QPT adversary that breaks the OWSG.
\fi

Let $(G,S)$ be a OW QGA.    
From it, we construct a weak OWSG $(\KeyGen,\allowbreak \StateGen,\allowbreak \Ver)$ as follows.
\begin{itemize}
    \item 
    $\KeyGen(1^\secp)\to k:$
    Run $[|s\rangle]\gets S(1^\secp)$ and $[g]\gets G(1^\secp)$.
    Output $k\coloneqq([|s\rangle],[g])$.
    \item 
    $\StateGen(k)\to|\phi_k\rangle:$
    Parse $k=([|s\rangle],[g])$.
    Output $|\phi_k\rangle\coloneqq |s\rangle\otimes g|s\rangle$.
    \item 
    $\Ver(k',|\phi_k\rangle)\to\top/\bot:$
    Parse $k'=([|s'\rangle],[g'])$.
    Apply $g'\otimes I$ on $|\phi_k\rangle$ and do the SWAP test between the two registers.
\end{itemize}
Assume that this is not weak one-way. Then, for any polynomial $p$, there exists a QPT $\cA$ and a polynomial $t$ such that
\ifnum\fullpage=1
\begin{align}
&\sum_{[|s\rangle],[g]}\Pr[[|s\rangle]\gets S(1^\secp)]\Pr[[g]\gets G(1^\secp)]
\sum_{[|s'\rangle],[g']}\Pr[([|s'\rangle],[g'])\gets\cA(1^\secp,(|s\rangle \otimes g|s\rangle)^{\otimes t})]
\cdot \frac{1+|\langle s|(g')^\dagger g|s\rangle|^2}{2}\\
&\ge 1-\frac{1}{p(\secp)},
\end{align}
\else
\begin{align}
&\sum_{[|s\rangle],[g]}\Pr[[|s\rangle]\gets S(1^\secp)]\Pr[[g]\gets G(1^\secp)] \nonumber \\ 
&\quad\left(
\sum_{[|s'\rangle],[g']}\Pr[([|s'\rangle],[g'])\gets\cA(1^\secp,(|s\rangle \otimes g|s\rangle)^{\otimes t})]
\cdot \frac{1+|\langle s|(g')^\dagger g|s\rangle|^2}{2} \right) \\ 
&\ge 1-\frac{1}{p(\secp)},
\end{align}
\fi
which means that
\ifnum\fullpage=1
\begin{align}
&\sum_{[|s\rangle],[g]}\Pr[[|s\rangle]\gets S(1^\secp)]\Pr[[g]\gets G(1^\secp)]
\sum_{[|s'\rangle],[g']}\Pr[([|s'\rangle],[g'])\gets\cA(1^\secp,(|s\rangle \otimes g|s\rangle)^{\otimes t})]
\cdot |\langle s|(g')^\dagger g|s\rangle|^2\\
&\ge 1-\frac{2}{p(\secp)}.
\end{align}
\else
\begin{align}
&\sum_{[|s\rangle],[g]}\Pr[[|s\rangle]\gets S(1^\secp)]\Pr[[g]\gets G(1^\secp)] \nonumber \\
&\quad\left(\sum_{[|s'\rangle],[g']}\Pr[([|s'\rangle],[g'])\gets\cA(1^\secp,(|s\rangle \otimes g|s\rangle)^{\otimes t})]
 \cdot |\langle s|(g')^\dagger g|s\rangle|^2 \right) \\
&\ge 1-\frac{2}{p(\secp)}.
\end{align}
\fi
It is clear that we can construct a QPT adversary that breaks the OW QGA from this $\cA$.

From Theorem 3.7 of \cite{TQC:MorYam24}, we obtain a pure OWSG $(\KeyGen',\StateGen',\Ver')$ from this weak OWSG.
Moreover, we can check that
$\Pr[\top\gets\Ver'(k,|\phi_k\rangle)]\ge1-\negl(\secp)$
is satisfied for all $k$.
Then, as is shown in Appendix B of \cite{TQC:MorYam24},
we can construct another OWSG with the canonical verification.
\end{proof}

\subsection{Pseudorandom QGAs}  \label{sec:QGA:PR}
We also introduce quantum analogue of pseudorandom group actions. 
We define three types of pseudorandomness of QGAs,
 which we call pseudorandom (PR), Haar-pseudorandom (Haar-PR), and DDH.  

\begin{definition}[PR QGAs]
We say that a QGA $(G,S)$ is \emph{pseudorandom (PR)} if the following two distributions are computationally indistinguishable for any polynomial $t$: 
\begin{align*}
D_{\subpr,0} &: 
\ket{s_0} \gets S, h \gets G;
\text{ return } (|s_0\rangle,h \ket{s_0})^{\otimes t} \\
D_{\subpr,1} &: 
\ket{s_0} \gets S, 
\ket{s} \gets \mu; 
\text{ return } (|s_0\rangle,\ket{s})^{\otimes t}. 
\end{align*}   
\end{definition}

We can show that the multiple samples are also pseudorandom.
%It is easy to check that the number of samples does not matter via a hybrid argument. We note that we can simulate $t$-samples from the Haar measure $\mu$ by using $t$-designs. 
\begin{lemma}\label{lem:PR_Q}
Let $(G,S)$ be a PR QGA.
Then the following two distributions are computationally indistinguishable for any polynomials $Q$ and $t$: 
\begin{align*}
D_{\subpr,0}' &: 
\ket{s_0} \gets S, \text{ for $q \in [Q]$ }  h_q \gets G;
\text{ return } \{(h_q \ket{s_0})^{\otimes t}\}_{q \in [Q]} \\
D_{\subpr,1}' &: 
\text{ for $q \in [Q]$ } \ket{s_q} \gets \mu; 
\text{ return } \{\ket{s_q}^{\otimes t} \}_{q \in [Q]}. 
\end{align*}
%\xagawa{If $\ket{s_0}$ is publicly-generable, then we can omit $\ket{s_0}$ from the above. We cannot remove $Q$ unless we can efficiently sample $\ket{x} \gets \mu$. This assumption with $Q=1$ directly implies PRSG.} 
%In other words, if PRSG generator $G(k)$ can be decomposed to two unitaries $S_j$ that prepares initial state $\ket{s_0} = S_j \ket{0^\lambda}$ and $G_k$ that output $\ket{\phi_k} = g_k \ket{s_0}$ ... then we have PRSG-like QGA.
%\xagawa{If we know $t$ in advance, then we can use a hybrid argument with $t$-designs to simulate $t$ copies $\ket{s_q}^{\otimes t}$ where $\ket{s_q} \gets \mu$. }
%\xagawa{In the proof of NR-like PRFSs, we don't need $\ket{s_0}$.}
\end{lemma}

\begin{proof}
Let $(G,S)$ be a PR QGA.
Define the distributions $H_j^{t,Q}$ for $j=0,\dots,Q$ as follows.
\begin{itemize}
\item
$\ket{s_0} \gets S$
    \item 
    For $q\in\{1,2,...,j\}$,
$ \ket{s_q} \gets \mu$.
\item 
For $q\in\{j+1,...,Q\}$,
$h_q \gets G$.
\item
Output
$\{\ket{s_q}^{\otimes t} \}_{q \in\{1,...,j\}}$
and
$\{(h_q \ket{s_0})^{\otimes t}\}_{q \in \{j+1,...,Q\}}$.
\end{itemize}
It is clear that $D_{\subpr,0}'=H_0^{t,Q}$
and $D_{\subpr,1}'=H_Q^{t,Q}$.
We claim that for any QPT adversary $\cA$, any polynomials $Q,t$, and any $j\in[Q]$
\begin{align*}
\left\lvert
\Pr[1\gets\cA(H_{j-1}^{t,Q})]
-\Pr[1\gets\cA(H_j^{t,Q})]
\right\rvert
\le\negl(\secp).
\end{align*}
To show this claim, assume that there exist a QPT $\cA$, polynomials $Q,t,p$, and $j\in[Q]$ such that
\begin{align*}
\left|    
\Pr[1\gets\cA(H_{j-1}^{t,Q})]
-\Pr[1\gets\cA(H_j^{t,Q})]
\right|
\ge\frac{1}{p(\secp)}
\end{align*}
for infinitely many $\secp$.
Then we can construct a QPT adversary $\cB$ that breaks the security of the PR QGA as follows:
\begin{enumerate}
\item 
The challenger $\cC$ chooses $|s_0\rangle\gets S$ and $b\gets\bit$.
    \item 
    If $b=0$, $\cC$ chooses $h\gets G$ and sends $(|s_0\rangle,h|s_0\rangle)^{\otimes Qt}$ to $\cB$.
    If $b=1$, $\cC$ chooses $|s\rangle\gets\mu$ and sends $(|s_0\rangle,|s\rangle)^{\otimes Qt}$ to $\cB$.
    \item 
    $\cB$ runs $\cA$ on input $\{|s_q\rangle^{\otimes t}\}_{q\in\{1,...,j-1\}}$, 
    the received state (i.e., $(h|s_0\rangle)^{\otimes t}$ or $|s\rangle^{\otimes t}$), and $\{(h_q|s_0\rangle)^{\otimes t}\}_{q\in\{j+1,...,Q\}}$,
    and outputs $\cA$'s output.
    Here, all $|s_q\rangle$ are $t$-designs and each $h_q\gets G$. (From \cref{lem:design-Haar}, we can replace Haar random states with $t$-design states.)
    Note that $\cB$ can efficiently generate the $\{(h_q|s_0\rangle)^{\otimes t}\}_{q\in\{j+1,...,Q\}}$
    because $\cB$ receives $|s_0\rangle^{\otimes Qt}$ from $\cC$.
\end{enumerate}
We have
\begin{align*}
\left|\Pr[1\gets \cB \mid b=0]-\Pr[1\gets\cA(H_{j-1}^{t,Q})]\right|&\le\negl(\secp),\\
\left|\Pr[1\gets \cB \mid b=1]-\Pr[1\gets\cA(H_{j}^{t,Q})]\right|&\le\negl(\secp).
\end{align*}
Therefore, we have
\begin{align*}
&\left|\Pr[1\gets\cB \mid b=0]
-\Pr[1\gets\cB \mid b=1]
\right|  \\
&\ge 
\left|\Pr[1\gets\cA(H_{j-1}^{t,Q})]
-\Pr[1\gets\cA(H_{j}^{t,Q})]\right|-\negl(\secp)\\
&\ge\frac{1}{p(\secp)}-\negl(\secp)
\end{align*}
for infinitely many $\secp$,
which means that the PR QGA is broken.
\end{proof}

It is obvious that PR QGAs directly imply PRSGs.
\begin{lemma}
\label{thm:PR->PRSG}
If PR QGAs exist, then PRSGs exist.    
\end{lemma}
\begin{proof} %[Proof of \cref{thm:PRQGA->PRSG}]
Let $(G,S)$ be a PR QGA. From it, we construct a PRSG, $(\KeyGen,\StateGen)$, as follows.
\begin{itemize}
\item 
$\KeyGen(1^\secp)\to k:$
Run $|s\rangle\gets S(1^\secp)$ and $g\gets G(1^\secp)$.
Output $k\coloneqq (|s\rangle,g)$.
    \item 
    $\StateGen(1^\secp,k)\to|\phi_k\rangle:$
    Parse $k=(|s\rangle,g)$.
    Output $|\phi_k\rangle\coloneqq g|s\rangle$.
\end{itemize}
It is clear that this satisfies the security of PRSGs.
\end{proof}

Next, we give a variant of pseudorandomness, which we call \emph{Haar-pseudorandom (Haar-PR)} because underlying states are generated according to the Haar measure. 
\begin{definition}[Haar-PR QGAs] \label{def:QGA:Haar-PR}
A QGA $(G,S)$ is called \emph{Haar-pseudorandom (Haar-PR)} if
the following two distributions are computationally indistinguishable for any polynomial $t$: 
\begin{align*}
D_{\subHaarPR,0} &: \ket{s} \gets \mu, h \gets G; \text{ return } (\ket{s}, h \ket{s})^{\otimes t} \\
D_{\subHaarPR,1} &: \ket{s} \gets \mu, \ket{s'} \gets \mu; \text{ return } (\ket{s}, \ket{s'})^{\otimes t}
\end{align*}
\end{definition}
\begin{remark}
Haar-PR QGAs can be considered as a computational version of 
the statistical construction of PRSGs in the common Haar state model~\cite{EPRINT:AnaGulLin24,CCS24}.
\end{remark}

We can show that the multiple samples are also pseudorandom as in the case of PR QGAs~\cref{lem:PR_Q}. 
\begin{lemma}\label{lem:HaarPR_Q}
Let $(G,S)$ be a Haar-PR QGA.
Then the following two distributions are computationally indistinguishable for any polynomials $Q$ and $t$: 
\begin{align*}
D_{\subHaarPR,0}' &: \text{ for $q \in [Q]$ } \ket{s_q} \gets \mu, h_q \gets G; \text{ return } \{(\ket{s_q}, h_q \ket{s_q})^{\otimes t}\}_{q \in [Q]} \\
D_{\subHaarPR,1}' &: \text{ for $q \in [Q]$ } \ket{s_q} \gets \mu, \ket{s'_q} \gets \mu; \text{ return } \{(\ket{s_q}, \ket{s'_q})^{\otimes t}\}_{q \in [Q]}. 
\end{align*}
\end{lemma}
%\xagawa{If we know $t$ in advance, then we can use a hybrid argument with $t$-designs to simulate $t$ copies $\ket{s_q}^{\otimes t}$ where $\ket{s_q} \gets \mu$. }
\begin{proof}
Let $(G,S)$ be a Haar-PR QGA.    
For each $j\in\{0,1,...,Q\}$, define the following distribution $H_j^{t,Q}$.
\begin{itemize}
    \item 
    For $q\in\{1,2,...,j\}$,
$\ket{s_q} \gets \mu$ and $\ket{s'_q} \gets \mu$.
    \item 
    For $q\in\{j+1,j+2,...,Q\}$,
$\ket{s_q} \gets \mu$ and $h_q \gets G$.
\item 
Output
$\{(\ket{s_q}, \ket{s'_q})^{\otimes t}\}_{q \in \{1,...,j\}}$ and
$\{(\ket{s_q}, h_q \ket{s_q})^{\otimes t}\}_{q \in \{j+1,...,Q\}}$. 
\end{itemize}
It is clear that $D_{\subHaarPR,0}'=H_0^{t,Q}$
and $D_{\subHaarPR,1}'=H_Q^{t,Q}$.
We claim that for any QPT adversary $\cA$, any polynomials $Q,t$, and any $j\in[Q]$
\begin{align*}
\left|    
\Pr[1\gets\cA(H_{j-1}^{t,Q})]
-\Pr[1\gets\cA(H_j^{t,Q})]
\right|
\le\negl(\secp).
\end{align*}
To show it, assume that there exist a QPT $\cA$, polynomials $Q,t,p$, and $j\in[Q]$ such that
\begin{align*}
\left|    
\Pr[1\gets\cA(H_{j-1}^{t,Q})]
-\Pr[1\gets\cA(H_j^{t,Q})]
\right|
\ge\frac{1}{p(\secp)}
\end{align*}
for infinitely many $\secp$.
Then we can construct a QPT adversary $\cB$ that breaks the security of the Haar-PR QGA as follows.
\begin{enumerate}
\item 
The challenger $\cC$ chooses $b\gets\bit$.
    \item 
    If $b=0$, $\cC$ chooses $|s\rangle\gets \mu$ and $h\gets G$.
    and sends $(|s\rangle,h|s\rangle)^{\otimes t}$ to $\cB$.
    If $b=1$, $\cC$ chooses $|s\rangle,|s'\rangle\gets\mu$,
    and sends $(|s\rangle,|s'\rangle)^{\otimes t}$ to $\cB$.
    \item 
    $\cB$ prepares $\{(|s_q\rangle,|s_q'\rangle)^{\otimes t})\}_{q\in\{1,...,j-1\}}$ and $\{|s_q\rangle,h_q|s_q\rangle)^{\otimes t}\}_{q\in\{j+1,...,Q\}}$ by using $t$-designs. 
    It then runs $\cA$ on input $\{(|s_q\rangle,|s_q'\rangle)^{\otimes t})\}_{q\in\{1,...,j-1\}}$, 
    the received state, and $\{|s_q\rangle,h_q|s_q\rangle)^{\otimes t}\}_{q\in\{j+1,...,Q\}}$,
    and outputs its output.
    Here, all $|s_q\rangle,|s_q'\rangle$ are $t$-designs.
\end{enumerate}
We have
\begin{align*}
\left|\Pr[1\gets \cB \mid b=0]-\Pr[1\gets\cA(H_{j-1}^{t,Q})]\right|&\le\negl(\secp),\\
\left|\Pr[1\gets \cB \mid b=1]-\Pr[1\gets\cA(H_{j}^{t,Q})]\right|&\le\negl(\secp).
\end{align*}
Therefore, we have
\begin{align*}
&\left|\Pr[1\gets\cB \mid b=0]
-\Pr[1\gets\cB \mid b=1]
\right| \\
&\ge 
\left|\Pr[1\gets\cA(H_{j-1}^{t,Q})]
-\Pr[1\gets\cA(H_{j}^{t,Q})]\right|-\negl(\secp)\\
&\ge\frac{1}{p(\secp)}-\negl(\secp)
\end{align*}
for infinitely many $\secp$,
which means that the Haar-PR QGA is broken.
\end{proof}

We also define a natural quantum analogue of \emph{Decisional Diffie-Hellman (DDH)}. 
\begin{definition}[DDH QGAs]  \label{def:QGA:qDDH}
A QGA $(G,S)$ is called \emph{Decisional Diffie-Hellman (DDH)} if
the following two distributions are computationally indistinguishable for any polynomials $Q$ and $t$: 
\begin{align*}
D_{\subDDH,0} &: |s_0\rangle \gets S, \tilde{g}, g \gets G; \text{ return } \{(\ket{s_0}, \tilde{g} \ket{s_0}, g \ket{s_0}, \tilde{g}g \ket{s_0})^{\otimes t}\}_{q \in [Q]} \\
D_{\subDDH,1} &: |s_0\rangle \gets S, \tilde{g}, g, h \gets G; \text{ return } \{(\ket{s_0}, \tilde{g} \ket{s_0}, g \ket{s_0}, h \ket{s_0})^{\otimes t}\}_{q \in [Q]}.
\end{align*}
\end{definition}
% \begin{remark}
% DDH QGAs can be considered as a quantum analogue of 
%  $Q$-diagonal pseudorandom group actions defined in~\cite{TCC:JQSY19}
%  and weak pseudorandom group actions defined in \cite{AC:ADMP20}. 
% They require computational indistinguishability between $\{(s_q,g \star s_q)\}_{q\in[Q]}$ and 
%  $\{(s_q,s'_q)\}_{q\in[Q]}$, where $g \gets G$ and $s_q,s'_q \gets S$ for $q\in[Q]$.
% If group action is regular~\cite{AC:ADMP20}, then the latter distribution is the same as $\{(s_q,h_q \star s_q)\}_{q\in[Q]}$, where $h_q \gets G$ and $s_q \gets S$ for $q\in[Q]$. 
% \mor{what do you mean here?}
% \end{remark}

We next give another variant of DDH, which we call \emph{Haar Decisional Diffie-Hellman (Haar-DDH)}. 
We will use it for the case that $G$ is non-commutative.
\begin{definition}[Haar-DDH QGAs]  \label{def:QGA:Haar-DDH}
A QGA $(G,S)$ is called \emph{Haar-Decisional Diffie-Hellman (Haar-DDH)} if
the following two distributions are computationally indistinguishable for any polynomials $Q$ and $t$: 
\begin{align*}
D_{\subHDDH,0} &: g \gets G, \text{ for $q \in [Q]$ } \ket{s_q} \gets \mu; \text{ return } \{(\ket{s_q}, g \ket{s_q})^{\otimes t}\}_{q \in [Q]} \\
D_{\subHDDH,1} &: \text{ for $q \in [Q]$ } \ket{s_q} \gets \mu, h_q \gets G; \text{ return } \{(\ket{s_q}, h_q \ket{s_q})^{\otimes t} \}_{q \in [Q]}. 
\end{align*}
\end{definition}

\subsection{Naor-Reingold QGAs}  \label{sec:QGA:NR}
We also introduce an assumption that is a key for our construction of
PRFSGs from pseudorandom QGAs. We dub it as \emph{Naor-Reingold (NR)} QGAs because this assumption will be used to show the pseudorandomness of the Naor-Reingold-style PRFSGs. 
\begin{definition}[NR QGAs] \label{def:QGA:NR}
A QGA $(G,S)$ is called \emph{Naor-Reingold (NR)} if
the following two distributions are computationally indistinguishable for any polynomial $Q$ and $t$: 
\begin{align*}
D_{\subNR,0} &: 
\ket{s_0} \gets S, \tilde{g} \gets G, \text{ for $q \in [Q]$ } g_q \gets G; 
\text{ return } \{(g_q \ket{s_0}, \tilde{g} g_q \ket{s_0})^{\otimes t}\}_{q \in [Q]} \\
D_{\subNR,1} &: 
\ket{s_0} \gets S, \text{ for $q \in [Q]$ } g_q \gets G, h_q \gets G; 
\text{ return } \{(g_q \ket{s_0}, h_q \ket{s_0})^{\otimes t}\}_{q \in [Q]}. 
\end{align*}
\end{definition}

If we consider classical group actions with special properties, the \NR-GA follows from the pseudorandomness of group actions.
(See \cref{sec:ClassicNRPRF} for the details.) 
% \footnote{Roughly speaking, it states that $(s_q, \tilde{g} \star s_q)_{q \in [Q]}$ and $(s_q, s_q')_{q \in [Q]}$ are computationally indistinguishable, where $s_q, s_q' \gets S$ and $\tilde{g}\gets G$~\cite{TCC:JQSY19,AC:ADMP20}.} 
In the case of QGA, we will face several problems since we cannot use algebraic structures; $G$ might not be a group, and $S$ is not closed. 
Fortunately, the NR property of QGA follows from its PR, Haar-PR, and Haar-DDH properties. 
\begin{lemma} \label{lem:PRetc->NR}
If a QGA is PR, Haar-PR, and Haar-DDH, %and $\ket{s_0}$ is publicly generable, 
then it is NR. 
\end{lemma}

\begin{proof} %[Proof of \cref{lem:PRetc->NR}]
To show the lemma, we introduce an intermediate distribution $D'_{\subNR}$ defined as follows: 
\begin{align*}
D'_{\subNR} &: 
\text{ for $q \in [Q]$ } \ket{s_q} \gets \mu, \tilde{h}_q \gets G; \text{ return } \{(\ket{s_q}, \tilde{h}_q \ket{s_q})^{\otimes t}\}_{q \in [Q]}. 
\end{align*}
We show that $D_{\subNR,0} \approx_c D'_{\subNR} \approx_c D_{\subNR,1}$ under our assumptions.

\begin{proof}[Proof of $D_{\subNR,0} \approx_c D'_{\subNR}$]
We first show that $D_{\subNR,0} \approx_c D'_{\subNR}$ if QGA is PR and Haar-DDH. 
Let us consider the following distribution:
\begin{align*}
D'_{\subNR,0} &:  \tilde{g} \gets G, \text{ for $q \in [Q]$ } \ket{s_q} \gets \mu; \text{ return } \{(\ket{s_q}, \tilde{g} \ket{s_q})^{\otimes t}\}_{q \in [Q]}. 
\end{align*}
As in \cref{claim:D_NR0=D'_NR0} below, we have $D_{\subNR,0} \approx_c D'_{\subNR,0}$ if QGA is PR. 
%and $\bar{D}_{\subNR}  \approx_c \bar{D}'_{\subNR}$.
%(See the following two claims.)
We also have $D'_{\subNR,0} \approx_c D'_{\subNR}$, which directly follows from the Haar-DDH assumption (\cref{def:QGA:Haar-DDH}). 
Hence, we have $D_{\subNR,0} \approx_c D'_{\subNR}$. 
\end{proof} 
\begin{claim} \label{claim:D_NR0=D'_NR0}
If QGA is PR, then $D_{\subNR,0} \approx_c D'_{\subNR,0}$. 
\end{claim}
\begin{proof}[Proof of \cref{claim:D_NR0=D'_NR0}] 
Assuming that QGA is PR, we know that $D'_{\subpr,0}$ and $D'_{\subpr,1}$ are computationally indistinguishable due to \cref{lem:PR_Q}. 
To show $D_{\subNR,0} \approx_c D'_{\subNR,0}$, 
 it is enough to show that we can efficiently convert samples from $D'_{\subpr,0}$ and $D'_{\subpr,1}$ into $D_{\subNR,0}$ and $D'_{\subNR,0}$, respectively, in an oblivious way. 
Suppose that we are given samples $\{\ket{y_q}^{\otimes 2t}\}_{q \in [Q]}$ from $D'_{\subpr,0}$ or $D'_{\subpr,1}$, 
where $\ket{y_q}=h_q\ket{s_0}$ for $D_{\subpr,0}$ 
and
$\ket{y_q}=\ket{s_q}$ for $D_{\subpr,1}$.
To prepare samples for $D_{\subNR,0}$ or $D'_{\subNR,0}$, we choose $\tilde{g} \gets G$ and 
 compute $(\ket{y_q}, \tilde{g} \ket{y_q})^{\otimes t}$ for $q \in [Q]$. 
If the input distribution is $D_{\subpr,0}$, then 
the output distribution is $D_{\subNR,0}$. 
On the other hand, if the input distribution is $D_{\subpr,1}$, then the output distribution is $D'_{\subNR,0}$. 
\end{proof}
% \begin{claim} \label{claim:barD_NR=barD'_NR}
% If QGA is PR, then $\bar{D}_{\subNR} \approx_c \bar{D}'_{\subNR}$.
% \end{claim}
% \begin{proof}[Proof of \cref{claim:barD_NR=barD'_NR}] 
% To show $D_{\subNR,0} \approx_c D'_{\subNR,0}$, 
%  we construct an efficient algorithm that converts samples from $D'_{\subpr,0}$ and $D'_{\subpr,1}$ in \cref{lem:PR_Q} and into samples from $\bar{D}_{\subNR}$ and $\bar{D}'_{\subNR}$, respectively. 
% Suppose that we are given samples $\{\ket{y_q}^{\otimes 2t}\}_{q \in [Q]}$ from $D'_{\subpr,0}$ or $D'_{\subpr,1}$, 
% where $\ket{y_q}=h_q\ket{s_0}$ for $D_{\subpr,0}$ 
% and $\ket{y_q}=\ket{s_q}$ for $D_{\subpr,1}$.
% To prepare samples for $\bar{D}_{\subNR}$ or $\bar{D}'_{\subNR}$, we choose $\tilde{h}_q \gets G$ and 
%  compute $(\ket{y_q}, \tilde{h}_q \ket{y_q})^{\otimes t}$ for $q \in [Q]$. 
% If the input distribution is $D_{\subpr,0}$, then 
% the output distribution is $\bar{D}_{\subNR}$.
% On the other hand, if the input distribution is $D_{\subpr,1}$, then the output distribution is $\bar{D}'_{\subNR}$ as we wanted. 
% \end{proof}

\begin{proof}[Proof of $D'_{\subNR} \approx_c D_{\subNR,1}$]
We then show $D'_{\subNR} \approx_c D_{\subNR,1}$ if QGA is PR and Haar-PR. 
% If $G$ is a group, then the two distributions coincide. 
% If not, we need to invoke the assumptions. 
We consider the following distribution: 
\begin{align*}
%\bar{D}'_{\subNR} &: \text{ for $q \in [Q]$ } \ket{s_q} \gets \mu, \tilde{h}_q \gets G; \text{ return } \{(\ket{s_q}, \tilde{h}_q \ket{s_q})^{\otimes t}\}_{q \in [Q]} \\
D'_{\subNR,1} &: \text{ for $q \in [Q]$ } \ket{s_q} \gets \mu, \ket{s'_q} \gets \mu; \text{ return } \{(\ket{s_q}, \ket{s'_q})^{\otimes t}\}_{q \in [Q]}.
\end{align*}
% Because QGA is PR, we have $\bar{D}_{\subNR} \approx_c \bar{D}'_{\subNR}$ (\cref{claim:barD_NR=barD'_NR}). 
The following claim (\cref{claim:D_NR1=D'_NR1}) shows that $D_{\subNR,1} \approx_c D'_{\subNR,1}$ if QGA is PR. 
\cref{lem:HaarPR_Q} shows that, if QGA is Harr-PR, then we have that $D'_{\subNR}  \approx_c D'_{\subNR,1}$. 
This completes the proof. 
\end{proof}
\begin{claim} \label{claim:D_NR1=D'_NR1}
If QGA is PR, then $D_{\subNR,1} \approx_c D'_{\subNR,1} $. 
\end{claim}
\begin{proof}[Proof of \cref{claim:D_NR1=D'_NR1}]
We construct an efficient quantum algorithm that converts samples from $D'_{\subpr,0}$ and $D'_{\subpr,1}$ in \cref{lem:PR_Q} and into samples from $D_{\subNR,1}$ and $D'_{\subNR,1}$, respectively. 
Suppose that we are given samples $\{\ket{y_q}^{\otimes t}\}_{q \in [2Q]}$ from $D'_{\subpr,0}$ or $D'_{\subpr,1}$, 
where $\ket{y_q}=h_q\ket{s_0}$ for $D_{\subpr,0}$ 
and $\ket{y_q}=\ket{s_q}$ for $D_{\subpr,1}$.
The converter outputs $\{ (\ket{y_{2q-1}}, \ket{y_{2q}})^{\otimes t} \}_{q \in [Q]}$ by rearranging samples. 
If the input distribution is $D_{\subpr,0}$, then 
the output distribution is $D_{\subNR,1}$.
On the other hand, if the input distribution is $D_{\subpr,1}$, then the output distribution is $D'_{\subNR}$ as we wanted. 
\end{proof}

Wrapping up the lemmas and claims, we have shown that $D_{\subNR,0} \approx_c D'_{\subNR} \approx_c D_{\subNR,1}$ in \cref{lem:PRetc->NR}. 
\end{proof}

If $G$ is \emph{commutative}, then we only need the DDH assumption as Boneh~et~al.~\cite{AC:BonKogWoo20}. 
\begin{lemma} \label{lem:com+DDH->NR}
Let $(G,S)$ be a QGA.
If $G$ is commutative and $(G,S)$ is DDH, 
then $(G,S)$ is NR. 
\end{lemma}
\begin{proof}
For $i = 0,\dots,Q$, we consider the following hybrid distributions $\bar{D}_i$ of $\{(|\phi_q\rangle, |\psi_q\rangle)^{\otimes t}\}_{q \in [Q]}$: 
\begin{itemize}
\item $\tilde{g} \gets G$ and $|s_0\rangle \gets S$. 
\item For $j = 1,\dots,i$, 
$|\phi_q\rangle \coloneqq g_q |s_0\rangle$ and $|\psi_q\rangle \coloneqq h_q |s_0\rangle$, where $g_q, h_q \gets G$. 
\item For $j = i+1,\dots,Q$, 
$|\phi_q\rangle \coloneqq g_q |s_0\rangle$ and $|\psi_q\rangle \coloneqq \tilde{g} h_q |s_0\rangle$, where $g_q, h_q \gets G$. 
\end{itemize}
By using the following claim, we have 
$\bar{D}_0 \approx_c \bar{D}_1 \approx_c \dots \approx_c \bar{D}_Q$ 
if the DDH assumption holds. 
\end{proof}
\begin{claim}
Let $(G,S)$ be a QGA.
If $G$ is commutative and $(G,S)$ is DDH, 
then, for $i = 1,\dots,Q$, 
$\bar{D}_{i-1} \approx_c \bar{D}_i$ holds. 
\end{claim}
\begin{proof}
Suppose that there exists $\cA$ distinguishing $\bar{D}_{i-1}$ from $\bar{D}_i$. 
We construct an adversary $\cB$ against the DDH assumption as follows:
\begin{itemize}
\item Given a sample $(|s_0\rangle, \tilde{g}|s_0\rangle, g|s_0\rangle, h|s_0\rangle)^{\otimes tQ}$, where $h = \tilde{g} g$ or random, $\cB$ prepares a sample $\{(|\phi_q\rangle, |\psi_q\rangle)^{\otimes t}\}_{q \in [Q]}$ as follows:  
\begin{itemize}
\item for $q=1,\dots,i-1$, take random $g_q,h_q \gets G$ and set $(|\phi_q\rangle, |\psi_q\rangle) \coloneqq (g_q |s_0\rangle, h_q |s_0\rangle)$; 
\item for $q=i$, set $(|\phi_q\rangle, |\psi_q\rangle) = (g|s_0\rangle, h|s_0\rangle)$;  
\item for $q=i+1,\dots,Q$, take random $g_q \gets G$ and set $(|\phi_q\rangle, |\psi_q\rangle) = (g_q |s_0\rangle, g_q \tilde{g} |s_0\rangle)$. 
\end{itemize}
\item It runs $\cA$ on input $\{(|\phi_q\rangle, |\psi_q\rangle)^{\otimes t}\}_{j \in [Q]}$ and outputs $\cA$'s decision. 
\end{itemize}
We note that, due to commutativity of $G$, the last $Q-i$ samples are equivalent to $(g_q |s_0\rangle, \tilde{g} g_q |s_0\rangle)$. 
If $h = \tilde{g}g$, then $(|\phi_i\rangle,|\psi_i\rangle) = (g |s_0\rangle, \tilde{g} g |s_0\rangle)$ and $\cB$ perfectly simulates the distribution $\bar{D}_{i-1}$ since $G$ is a group and the distribution of 
If $h$ is random, then $\cB$ perfectly simulates the distribution $\bar{D}_{i}$. 
Thus, $\cB$'s advantage is equivalent to $\cA$'s advantage distinguishing $\bar{D}_{i-1}$ and $\bar{D}_i$. 
\end{proof}

% \subsection{LHS assumption}
% \mor{remove this subsection?}
% \cite{AC:ADMP20} introduced this assumption, and constructed KDM-secure SKE. This can be done in our QGA.

% This assumption says that 
% $\{\ket{s_q}, \vec{g}_q, g_{q,1}^{s_1} \dots g_{q,\ell}^{s_\ell} \ket{s_q}\}_{q \in [Q]}$
% is computationally indistinguishable from 
% $\{\ket{s_q}, \vec{g}_q, \ket{s'_q}\}_{q \in [Q]}$. 

% \xagawa{Reading the proof of \cite{AC:ADMP20}, they used the fact that $G$ is a group. If $G$ is a group, then $M \gets G^{n \times n}$ is statistically indistinguishable $M - h \cdot I$ or $M + h \cdot I$.}

\section{Construction of (Classical-Query) PRFSGs}  \label{sec:PRFSG}

% * 分布1: $\ket{s_0}, (g_q \ket{s_0}, \tilde{g} g_q \ket{s_0})_{q =1,\dots,Q}$  (各々適当な数コピーが必要)
% * 分布2: $\ket{s_0}, (g_q \ket{s_0}, \tilde{h}_q g_q \ket{s_0})_{q =1,\dots,Q}$ 
% * 分布3: $\ket{s_0}, (g_q \ket{s_0}, h_q \ket{s_0})_{q =1,\dots,Q}$ 
% として、分布1 と 分布2の計算量的識別不可能性と、分布2と分布3の計算量的識別不可能性
% $\mathbb{G}$が群の場合は、分布2 = 分布3が言える。
% これらを使うとハイブリッドの$H_j$と$H_{j+1}$が識別不可能になる。
% また$j=\ell$まで進んだあと、Haarランダムなものと識別不可能にしたいので、
% * 分布4: $\ket{s_0}, h_1\ket{s_0},\dots,h_Q \ket{s_0}$ (各々適当な数コピー)
% * 分布5: $\ket{s_0}, \ket{s_1},\dots,\ket{x_q}$ ($\ket{x_q}$ for $q=1,\dots,Q$はHaarランダム)
% として分布4と分布5の識別不可能性も必要。
% 注：分布4と分布5の識別不可能性は, $\ket{s_0}$が公開情報の場合、$Q=1$で言えればハイブリッド法で示せる。
% 分布4 $\ket{s_0}, h_1\ket{s_0}$ vs. 分布5 $\ket{s_0}, \ket{s_1}$ はある意味PRSGの特殊版。
% PRSGで$U_k \ket{0^n}$が$U_k' \ket{s_0}$と書ける必要がある。
% 分布4 vs. 分布5から 1-2, 2-3を示せるか？
% * 分布1': $\ket{s_0}, (\ket{x_q}, \tilde{g} \ket{x_q})_{q =1,\dots,Q}$  (各々適当な数コピーが必要)
% * 分布2': $\ket{s_0}, (\ket{x_q}, \tilde{h}_q \ket{x_q})_{q =1,\dots,Q}$ 
% * 分布3': $\ket{s_0}, (\ket{x_q}, \ket{x_q'})_{q =1,\dots,Q}$ 
% 分布4 vs. 分布5が言えると, 分布i vs 分布i' が言える。
% なので、1' - 2' を仮定にするのはどうか？
% 2' - 3'は、改めて分布4 vs. 分布5の強い版を適用すれば良い（$(\ket{x_q}, h_q\ket{x_q})$ vs. $(\ket{x_q},\ket{x_q'})$）
% UPSGを示す場合は、$j=\ell$まで進めた後、敵が指定してくるタグに対応する状態が$h_{Q+1} \ket{s_0}$になっているので、それがunpredictableという仮定にすれば良い。これはほぼOWSGのようなもの。

\subsection{Construction}
We construct a Naor-Reingold-style PRFSG from QGA (that is secure against classical queries) in this section. 
\begin{theorem} \label{thm:PRNR->PRFSG}
Let $(G,S)$ be a QGA. 
If $(G,S)$ is PR and NR, then 
the following $(\KeyGen,\StateGen)$ is a PRFSG whose input space is $\{0,1\}^\ell$. 
\begin{itemize}
\item 
$\KeyGen(1^\secp)\to k:$ 
Sample $g_0,g_1,...,g_\ell\gets G$ and $|s_0\rangle\gets S$.
Output $k\coloneqq (g_0,g_1,...,g_\ell,|s_0\rangle)$.
(Note that $|s_0\rangle$ here is not a physical quantum state but its classical description.)
\item
$\StateGen(k,x)\to |\phi_k(x)\rangle:$ 
Parse
$k= (g_0,g_1,...,g_\ell,|s_0\rangle)$
and $x = (x[1],\dots,x[\ell]) \in \bit^\ell$. 
Output
\begin{align}
|\phi_k(x)\rangle\coloneqq g_\ell^{x[\ell]}g_{\ell-1}^{x[\ell-1]} \cdots g_1^{x[1]}g_0 \ket{s_0}. 
\end{align}
\end{itemize}
\end{theorem}

From \cref{lem:PRetc->NR} and \cref{lem:com+DDH->NR}, we obtain the following corollaries.
\begin{corollary}
If a QGA is PR, Haar-PR, and Haar-DDH, then
the above construction is a PRFSG.
\end{corollary}
\begin{corollary}
If a QGA is DDH and $G$ is commutative, then 
the above construction is a PRFSG.
\end{corollary}

To show \cref{thm:PRNR->PRFSG}, we define three games $\mathsf{Real}$, $\mathsf{Hybrid}$, and $\mathsf{Ideal}$ defined as follows: 
\begin{itemize}
\item $\mathsf{Real}$: This is the PRFSG security game whose oracle is $\StateGen(k,\cdot)$. That is, 
the challenger $\cC$ chooses $k \gets \KeyGen(1^\secp)$ and runs $\cA$ with the oracle $\StateGen(k,\cdot)$, which takes $x \in \{0,1\}^\ell$ as input and returns $\ket{\phi_k(x)}$.
$\cC$ outputs $\cA$'s decision. 
\item $\mathsf{Hybrid}$: This is the PRFSG security game whose oracle is defined as follows: On query $x$, if it is not queried before, then the oracle samples $h_{x} \gets G$ and returns $h_{x} \ket{s_0}$; otherwise, it returns stored $h_{x} \ket{s_0}$. 
\item $\mathsf{Ideal}$: This is the PRFSG security game whose oracle is $\cO_{\subHaar}$ defined as follows: On query $x$, 
if it is not queried before, then $\cO_{\subHaar}$ samples $\ket{s_{x}} \gets \mu$ and returns $\ket{s_{x}}$; otherwise, it returns $\ket{s_{x}}$. 
\end{itemize}
\cref{lem:PRNR->PRFSG:NR->Real-Hybrid} below shows that $\mathsf{Real}$ is computationally indistinguishable from $\mathsf{Hybrid}$ if QGA is NR. The proof is obtained by following the proofs in Naor and Reingold~\cite{JACM:NaoRei04} and Alamati~et~al.~\cite{AC:ADMP20}. 
% (Boneh, Montgomery, and Raghunathan~\cite{CCS:BonMonRag10} ... 
It is easy to show that $\mathsf{Hybrid}$ and $\mathsf{Ideal}$ are computationally indistinguishable if QGA is PR (\cref{lem:PRNR->PRFSG:PR->Hybrid-Ideal}). 
Thus, we obtain \cref{thm:PRNR->PRFSG}. The lemmas follow. 

\begin{lemma} \label{lem:PRNR->PRFSG:NR->Real-Hybrid} 
$\mathsf{Real}$ and $\mathsf{Hybrid}$ are computationally indistinguishable if QGA is \NR.
\end{lemma}
\begin{proof}
We define the following hybrid games $\mathsf{Game}_j$ for $j = 0,\dots,\ell$: 
The challenger samples $\ket{s_0} \gets S$, $g_0 \gets G$, and $g_i \gets G$ for $i \in [j+1,\ell]$. 
Let $q$-th adversary's query be $x_q= (x_q[1],\dots,x_q[\ell]) \in \{0,1\}^\ell$. 
The challenger returns a quantum state $\ket{f_{j,q}}$ defined as follows: 
\begin{enumerate}
\item If $j = 0$, then the challenger let $\ket{y_q} = g_0 \ket{s_0}$. 
\item Otherwise
\begin{itemize}
    \item if there exists $q' < q$ satisfying $(x_q[1],\dots,x_q[j]) = (x_{q'}[1],\dots,x_{q'}[j])$, then $\ket{y_q} = \ket{y_{q'}}$; 
    \item otherwise, it samples $g_q \gets G$ and $\ket{y_q} \coloneqq g_q \ket{s_0}$. 
\end{itemize}
\item Return $\ket{f_{j,q}} \coloneqq g_\ell^{x_q[\ell]} \cdots g_{j+1}^{x_q[j+1]} \ket{y_q}$. 
\end{enumerate}
It is easy to verify $\mathsf{Game}_0$ and $\mathsf{Game}_\ell$ are $\mathsf{Real}$ and $\mathsf{Hybrid}$, respectively. 
The following claim shows $\mathsf{Game}_0 \approx_c \mathsf{Game}_1 \approx_c \dots \approx_c \mathsf{Game}_\ell$ if the QGA is \NR. 
\end{proof}
\begin{claim} \label{claim:PRFS-hybrid}
For $j = 0,\dots,\ell-1$, 
$\mathsf{Game}_j$ and $\mathsf{Game}_{j+1}$ are computationally indistinguishable 
if the QGA is \NR. 
\end{claim}
\begin{proof}
The definition of \NR QGAs implies that $D_{\subNR,0} \approx_c D_{\subNR,1}$ under our hypothesis. 
Thus, it is enough to construct a reduction algorithm $\mathcal{B}$ distinguishing $D_{\subNR,0}$ and $D_{\subNR,1}$ by using an adversary $\mathcal{A}$ distinguishing $\mathsf{Game}_j$ and $\mathsf{Game}_{j+1}$. 
Our reduction algorithm is defined as follows: 
\begin{enumerate}
\item $\mathcal{B}$ is given $\{(\ket{y_q}, \ket{z_q})^{\otimes Q}\}_{q \in [Q]}$, where $(\ket{y_q},\ket{z_q}) = (g_q \ket{s_0}, \tilde{g} g_q \ket{s_0})$ in $D_{\subNR,0}$ or $(g_q \ket{s_0}, h_q \ket{s_0})$ in $D_{\subNR,1}$. It prepares $g_{j+2},\dots,g_\ell \gets G$. It initializes $c = 1$. 
\item Receiving a $q$-th query $x_q$ from $\mathcal{A}$, it checks if there exists $q' < q$ satisfying $(x_q[1],\dots,x_q[j]) = (x_{q'}[1],\dots,x_{q'}[j])$. If so, it uses previously-defined consistent quantum states, that is, sets $\ket{\tilde{y}_q} \coloneqq \ket{\tilde{y}_{q'}}$ and $\ket{\tilde{z}_q} \coloneqq \ket{\tilde{z}_{q'}}$. Otherwise, it picks new quantum states from the pool, that is, sets  $\ket{\tilde{y}_q} \coloneqq \ket{y_{c}}$ and $\ket{\tilde{z}_q} = \ket{z_{c}}$ and increments $c$. It then answers 
\[
\ket{f_{j,q}} \coloneqq 
\begin{cases}
g_\ell^{x_q[\ell]} \cdots g_{j+2}^{x_q[j+2]} \ket{\tilde{y}_q} & \text{if $x_q[j+1] = 0$} \\ 
g_\ell^{x_q[\ell]} \cdots g_{j+2}^{x_q[j+2]} \ket{\tilde{z}_q} & \text{if $x_q[j+1] = 1$}. 
\end{cases}
\]
\item $\mathcal{B}$ outputs $\mathcal{A}$'s decision. 
\end{enumerate}
If the given samples follow $D_{\subNR,0}$, then the above simulation perfectly simulates $\mathsf{Game}_j$ by considering $\tilde{g}$ as $g_{j+1}$. 
If the given samples follow $D_{\subNR,1}$, then the above simulation perfectly simulates $\mathsf{Game}_{j+1}$ since $g_q$'s and $h_q$'s in $D_{\subNR,1}$ are chosen independently. 
Thus, the claim follows. 
\end{proof}

\begin{lemma} \label{lem:PRNR->PRFSG:PR->Hybrid-Ideal}
$\mathsf{Hybrid}$ and $\mathsf{Ideal}$ are computationally indistinguishable if QGA is PR. 
\end{lemma}
\begin{proof}
This lemma immediately follows from \cref{lem:PR_Q}. 
\end{proof}

\subsection{On Quantum-Query PRFSGs} \label{sec:PRFSG:quantum-query}

%Maybe we can show if a classical-query PRFSG exists, then there is a classical-query PRFSG but it is not quantum-query PRFSG. 
As remarked in~\cref{remark:PRFSGs:Query}, we only consider classical-query PRFSGs (\cref{def:PRFSGs}). 
Currently, we fail to show either positive results that our NR-type PRFSG is secure against quantum queries 
or negative results, that is, the separation of classical-query and quantum-query PRFSGs. 
%that classical-query PRFSG implies the separation of classical-query and quantum-query PRFSGs. 
We discuss barriers for positive or negative results in the following. 

%\subsubsection{Barriers for positive results} 
\paragraph{Barriers for positive results.} 
We tried to show our NR-type PRFSG is secure against \emph{quantum} queries, but we faced some problems. 
For example, we can consider the quantum-query version of $\mathsf{Hybrid}$ and $\mathsf{Ideal}$. Can we show the computational indistinguishability between them from PR QGAs?

%\paragraph{Comments on small-domain separations:}
For simplicity, we consider the case of small input space $\bit^\ell$ with $\ell = O(\log(\lambda))$.\footnote{If $\ell = \omega(\log(\lambda))$, we then invoke the small-range distribution argument in Zhandry~\cite{C:Zhandry12}.}
We identify $\bit^\ell$ with $[2^\ell]$. 
If we consider PRFs, then classical-query PRFs against a quantum adversary is also quantum-query PRFs. 
This is because a classical-query adversary can ask all inputs $1,\dots,{2^\ell}$ to its oracle and simulate answers on quantum queries. 
However, in the case of PRFSGs, we do not have such implications. 
A classical-query adversary can obtain all states $\ket{s_x}$, which is $\ket{\phi_k(x)}$ or $\ket{\psi_x}$, for all $x \in [2^\ell]$. 
How can we simulate the answers to quantum queries? 

For example, if we implement a mapping $\ket{x} \mapsto \ket{x}\ket{s_x}$ by using a controlled swap, then the results will be as follows: 
\begin{align*}
&\ket{x} \ket{0^n} \otimes (\ket{s_1}^{\otimes Q} \otimes \dots \otimes \ket{s_{2^\ell}}^{\otimes Q}) \\
&\quad\mapsto_{\text{controlled swap}} \ket{x} \ket{s_{x}} \otimes (\ket{s_1}^{\otimes Q} \otimes \dots \otimes 
\ket{0^n} \ket{s_{x}}^{\otimes (Q-1)}
\otimes \dots 
\otimes \ket{s_{2^\ell}}^{\otimes Q}). 
\end{align*}
Unfortunately, these operations produce an entanglement between the states for the adversary and the reduction algorithm $\cB$
 and our attempt fails.

\paragraph{Barriers for negative results.} 
Zhandry showed that if there exists a classical-query PRF, then there is a classical-query PRF insecure against quantum-query attacks~\cite[Theorem 3.1]{FOCS:Zhandry12}. 
One would consider we can show its analogy for PRFSGs by mimicking his proof. 
Unfortunately, this strategy does not work because of the following reasons. 

We briefly review Zhandry's strategy: 
Suppose that there exists a PRF $\mathsf{PRF}$ whose input space is $[N]$, where $N = 2^{\omega(\log(\lambda))}$. (Otherwise, there is no separation as we explained in the above.)
We then define a new PRF whose input space is $[N']$, where $N'$ is a power of $2$ larger than $4 N^2$. 
The new key is a pair of the original key $k$ and a random prime $a \in (N/2,N]$. 
The new PRF is defined as $\mathsf{PRF}' \colon ((k,a) , x) \mapsto \mathsf{PRF}(k, x \bmod{a})$, which has a secret period $a$. 
Zhandry then showed that 1) if $\mathsf{PRF}$ is classical-query secure, then $\mathsf{PRF}'$ is also and 
2) if $\mathsf{PRF}$ is quantum-query secure, then $\mathsf{PRF}'$ is \emph{not}, which implies there exists a PRF that is classical-query secure but quantum-query insecure. 
To show 2), Zhandry constructed a quantum-query adversary breaking the security of $\mathsf{PRF}'$. 
Roughly speaking, using the period-finding algorithm in Boneh and Lipton~\cite{C:BonLip95}, the adversary can find a period $a$ in polynomial time with a probability of at least $1/2$.
The Boneh-Lipton period-finding algorithm uses a sufficiently large number $W$ and prepare a quantum state $\frac{1}{\sqrt{W}} \sum_{x,s \in \Z_W} \exp(2\pi i xs/W) \ket{s}\ket{\mathsf{PRF}'_{k,a}(x)} = \frac{1}{\sqrt{W}} \sum_{x,s \in \Z_W} \exp(2\pi i xs/W) \ket{s}\ket{\mathsf{PRF}_{k}(x \bmod{a})}$. 
We note that the success probability of the algorithm strongly depends on the distinctness of $\mathsf{PRF}(k,0),\dots,\mathsf{PRF}(k,a-1)$ and good measurements. The distinctness follows from the quantum-query security of $\mathsf{PRF}$. The measurement is done by using the computational basis. 

In the case of PRFSGs, we can construct an analogue of a new PRF in the same way.
While the distinctness follows from the quantum-query security of $\mathsf{PRFSG}$, we fail to give the measurement to distinguish $a$ independent samples $|\psi_1\rangle,\dots,|\psi_{a}\rangle$ from the Haar measure. 
\section{Candidates of QGA} \label{sec:candidates}
In this section, we provide some examples of candidate constructions of QGAs.
The one is taken from random quantum circuits. The other candidates are inspired by Instantaneous quantum polynomial (IQP) circuits~\cite{BreMonShe16,BreJozShe10}.
Those ones feature a commutativity of two unitaries taken by $G$. 

\begin{definition}[Candidate 1: Random circuit QGA]
We define 
\begin{itemize}
\item 
$G(1^\secp)\to [g]:$ 
Output a random quantum circuit $g$. 
\item
$S(1^\secp)\to [\ket{s}]:$
Output $\ket{s} \coloneqq \ket{0^\secp}$.  
\end{itemize}
\end{definition}
Random quantum circuits are conjectured to be PRUs~\cite{C:AnaQiaYue22},
while PR, Haar-PR, and Haar-DDH QGA seem to be incomparable with PRUs. 

Inspired by IQP, we conjecture that the following two QGAs are PR, Haar-PR, and Haar-DDH. 
\begin{definition}[Candidate 2: IQP QGA with random $Z$-diagonal circuit]
We define an IQP QGA $(G, S)$ as follows: 
\begin{itemize}
\item 
$G(1^\secp)\to [g]:$ 
Take a random $Z$-diagonal circuit $D$~\footnote{
E.g., a random circuit with gates $\{T,\mathit{CS}\}$~\cite{BreMonShe16}.
}  and output a description $[D]$. 
This defines $g = H^{\otimes \secp} \cdot D \cdot H^{\otimes \secp}$. 
\item
$S(1^\secp)\to [\ket{s}]:$
Output $\ket{s} \coloneqq \ket{0^\secp}$. 
\end{itemize}
\end{definition}
\begin{definition}[Candidate 3: IQP QGA with random sparse polynomials]
Let $d \in [1,\secp]$ and polynomial $w = w(\secp)$. 
Let $\cD_{d,w}$ be 
\[
 \cD_{d,w} \coloneqq \left\{
 \begin{matrix}
 D \colon \ket{x_1,\dots,x_\secp} \to (-1)^{f(x_1,\dots,x_\secp)}\ket{x_1,\dots,x_\secp} \\
 \mid f \in \mathbb{F}_2[x_1,\dots,x_\secp], \deg(f) \leq d, \mathrm{term}(f) \leq w 
 \end{matrix}
 \right\},
\]
where $\deg$ is a total degree of $f$ and $\mathrm{term}$ is a number of terms of $f$. 
We then define a set of IQP circuits with respect to $\cD$: 
\[
 \cG_{d,w} \coloneqq \{ H^{\otimes \secp} \cdot D \cdot H^{\otimes \secp} \mid D \in \cD_{d,w} \}. 
\]
We define an IQP QGA $(G, S)$ as follows: 
\begin{itemize}
\item 
$G(1^\secp)\to [g]:$ 
Take a random sample $D \gets \cD_{d,w}$ and output a description $[D]$. 
This defines $g = H^{\otimes \secp} \cdot D \cdot H^{\otimes \secp}$. 
\item
$S(1^\secp)\to [\ket{s}]:$
Output $\ket{s} \coloneqq \ket{0^\secp}$. 
\end{itemize}
\end{definition} 
%While it is known the output distribution of noisy IQP circuits is simulated in classical polynomial-time when the depth of $D$ is constant~\cite{RajWatLiu24} % https://arxiv.org/abs/2403.14607
% , it is believed hard if the depth of $D$ is polynomial. 
Candidates 2 and 3 differ the way to choose the $Z$-digaonal circuit $D$. 
We note that  $g h = h g$ holds for any $g = H^{\otimes \secp} D_g H^{\otimes 
\secp}$ and $h = H^{\otimes \secp} D_h H^{\otimes \secp}$ chosen by $G$ in both candidates. 
% $\cG_d$. 
% Thus, $G$ in candidate 3 has commutativity while that in candidate 2 not. 

IQP random circuits are not PRUs. In fact, if the adversary queries $H^{\otimes \secp}|0^\secp\rangle$ to the oracle,
the oracle always returns 
$H^{\otimes \secp}|0^\secp\rangle$ if the oracle is the IQP oracle, but such probability is exponentially small if the oracle is the Haar random unitary oracle.\footnote{We thank Shogo Yamada for pointing out it.} 
% \xagawa{Since there might be a constant term, it will be $(-1)^{f(0,\dots,0)} H^{\otimes \secp} |0^\secp\rangle$.}
Currently, we do not know that IQP random circuits are PR QGAs; 
It is shown that the state $\sum_x (-1)^{f(x)}|x\rangle$ with random $f$ is Haar random~\cite{C:BraShm20}.
We do not find any evidence that IQP random circuits are not Haar-PR or DDH QGAs.

\ifnum\anonymous=1
\else
\vspace{1cm}
\textbf{Acknowledgements.}
TM is supported by
JST CREST JPMJCR23I3,
JST Moonshot R\verb|&|D JPMJMS2061-5-1-1, 
JST FOREST, 
MEXT QLEAP, 
the Grant-in Aid for Transformative Research Areas (A) 21H05183,
and 
the Grant-in-Aid for Scientific Research (A) No.22H00522.
\fi

\ifnum\submission=0
\bibliographystyle{alpha} 
\else
\bibliographystyle{splncs04}
\fi
\bibliography{abbrev3,crypto,reference,text}

\appendix

\clearpage

\begin{center}
{\Huge
%\textbf{Supplementary Materials}}
\textbf{Appendix}}
\end{center}
%\section{Miscellaneous Results} \label{sec:miscellaneous}

% \subsection{OW QGAs imply OWSGs} \label{sec:OWQGA->OWSG}

% \xagawa{Space for LNCS}

% \subsection{PR QGAs imply PRSGs} \label{sec:PRQGA->PRSG}

% \xagawa{Space for LNCS}

% \subsection{HaarPR + Haar-DDH implies SKE} \label{sec:HaarPR+Haar-DDH->SKE}

\section{Haar-PR and Haar-Haar-DDH imply SKE} \label{sec:HaarPR+Haar-DDH->SKE} 

\paragraph{Preliminaries.}

We first review the definitions of SKE. 
Our IND-CPA definition is ``real-or-fixed'' style\footnote{An adversary has access to the oracle that takes an input $x$ and returns an encryption of $x$ or a fixed element $x'$ depending on fixed $b \in \bit$ and guesses $b$.}.
It is easy to show that this security implies left-or-right IND-CPA security\footnote{An adversary has access to the oracle that takes two inputs $x_0,x_1$ and returns an encryption of $x_b$ with fixed $b \in \bit$ and guesses $b$.}.
via a hybrid argument. 
It is also easy to show that left-or-right IND-CPA security implies find-then-guess IND-CPA securities\footnote{An adversary has access to the oracle that takes an input $x$ and returns an encryption of $x$ and distinguish an encryption of $x_0^*$ or $x_1^*$.}
 defined in \cite{C:BroJef15,ICITS:ABFGSJ16,EPRINT:MorYamYam24}
 by following the proof in~\cite{FOCS:BDJR97}.

\begin{definition}[Classical-message SKE]
A symmetric-key encryption (SKE) scheme for $\ell$-bit classical messages consists of three algorithms $(\KeyGen, \Enc,\Dec)$ such that 
\begin{itemize}
\item $\KeyGen(1^\secp) \to K:$
This is a QPT algorithm that takes the security parameter $1^\secp$ as input and outputs a \emph{classical} secret key $K$.
\item $\Enc(K,b) \to \ct:$
This is a QPT algorithm that takes $K$ and a message $b \in \bit^\ell$ and outputs a \emph{quantum} ciphertext $\ct$. 
\item $\Dec(K,\ct) \to b' / \bot:$
This is a QPT algorithm that takes $K$ and a quantum ciphertext $\ct$ and outputs a classical message $b' \in \bit^\ell$ or a rejection symbol $\bot$. 
\end{itemize}
\end{definition}

\begin{definition}[IND-CPA-secure $1$-bit SKE]
We say that a SKE scheme $(\KeyGen,\Enc,\Dec)$ is \emph{IND-CPA-secure $1$-bit SKE} 
if it satisfies the following two properties: 
\begin{itemize}
\item Correctness: 
We have 
\begin{align*}
&\Pr_{K \gets \KeyGen(1^\secp)}[\Dec(K,\Enc(K,0)) = 0] = 1, \\
&\Pr_{K \gets \KeyGen(1^\secp)}[\Dec(K,\Enc(K,1)) = 1] \geq 1/5.
\end{align*}
\item IND-CPA security: We call SKE \emph{indistinguishable against chosen-plaintext attacks (IND-CPA secure)} if the following holds: 
For any QPT adversary $\cA$, 
\[
\left\lvert 
\Pr_{K \gets \KeyGen(1^\secp)}[1 \gets \cA^{\Enc(K,\cdot)}(1^\secp)] 
- 
\Pr_{K \gets \KeyGen(1^\secp)}[1 \gets \cA^{\Enc(K,1)}(1^\secp)] 
\right\rvert
\leq \negl(\secp),
\]
where $\cA$ queries the encryption oracles $\Enc(K,\cdot)$ or $\Enc(K,1)$ only classically. 
\end{itemize}
\end{definition}

\begin{definition}[IND-CPA-secure multi-bit SKE]
We say that a SKE scheme $(\KeyGen,\Enc,\Dec)$ is \emph{IND-CPA-secure multi-bit SKE} 
if it satisfies the following two properties: 
\begin{itemize}
\item Correctness: for any $m \in \bit^\ell$, 
\[
\Pr_{K \gets \KeyGen(1^\secp)}[\Dec(K,\Enc(K,m)) = m] \geq 1 - \negl(\secp).
\]
\item IND-CPA security: We call SKE \emph{indistinguishable against chosen-plaintext attacks (IND-CPA secure)} if the following holds: 
For any QPT adversary $\cA$, 
\[
\left\lvert 
\Pr_{K \gets \KeyGen(1^\secp)}[1 \gets \cA^{\Enc(K,\cdot)}(1^\secp)] 
- 
\Pr_{K \gets \KeyGen(1^\secp)}[1 \gets \cA^{\Enc(K,1^\ell)}(1^\secp)] 
\right\rvert
\leq \negl(\secp),
\]
where $\cA$ queries the encryption oracles $\Enc(K,\cdot)$ or $\Enc(K,1^\ell)$ only classically. 
\end{itemize}
\end{definition}
\begin{lemma}
Suppose that there exists an IND-CPA-secure $1$-bit SKE $(\KeyGen,\allowbreak \Enc,\allowbreak \Dec)$. 
Let $t = t(\secp) = \omega(\log(\secp))$ and  $\ell = \ell(\secp)$ be polynomials. 
We define the following new SKE scheme: 
\begin{itemize}
\item $\KeyGen':$
Run $\KeyGen(1^\secp)$ $t\ell$-times and obtain $K_1,\dots,K_{t \ell}$. Output $K' = (K_1,\dots,K_{t \ell})$.
\item $\Enc'(K',m):$ 
Let $m = (m_1,\dots,m_\ell) \in \bit^\ell$. 
For $i \in [\ell]$ and $j \in [t]$, generate $\ct_{(i-1)t + j} \gets \Enc(K_{(i-1)t + j},m_i)$. 
Output $\ct' = (\ct_1,\dots,\ct_{t \ell})$. 
\item $\Dec'(K',\ct'):$
For $i \in [\ell]$ and $j \in [t]$, let $m_{(i-1)t + j} \gets \Dec(K_{(i-1)t+j},\ct_{(i-1)t+j})$. 
For $i \in [\ell]$, if $m_{(i-1)t+j}=0$ for all $j\in [t]$, then set $m'_i = 0$; otherwise, set $m'_i = 1$. 
Output $m' = (m'_1,\dots,m'_\ell)$. 
\end{itemize}
This new SKE $(\KeyGen',\Enc',\Dec')$ is IND-CPA-secure multi-bit SKE with plaintext space $\bit^\ell$. 
\end{lemma}
\begin{proof}
The correctness of the new multi-bit SKE follows from that of the underlying $1$-bit SKE: 
If $m_i = 0$, then the new decryption algorithm always outputs $m'_i = 0$. 
If $m_i = 1$, then the new decryption algorithm outputs $m'_i = 1$ with probability at least 
$1 - (1 - 1/5)^t = 1 - (4/5)^{\omega(\log(\secp))} = 1 - \delta(\secp)$
for some negligible function $\delta(\secp)$. 
Thus, for any $m \in \bit^\ell$, we have
\begin{align*}
\Pr_{K' \gets \KeyGen'(1^\secp)}[\Dec'(K',\Enc'(K',m)) = m] 
&\geq 1 - \Pr_{K' \gets \KeyGen'(1^\secp)}[\exists i \in [\ell]: 1 = m_i \neq m_i' = 0] \\
&\geq 1 - \ell \cdot \delta(\secp) \\
&= 1 - \negl(\secp)
\end{align*}
as we wanted. 

The IND-CPA security of the new multi-bit SKE immediately follows from that of the underlying $1$-bit SKE via a hybrid argument. 
\end{proof}

\paragraph{Construction.}
We can construct a simple IND-CPA-secure $1$-bit SKE from the HaarPR and Haar-DDH assumptions: 
\begin{theorem} \label{thm:HaarPR+Haar-DDH->1-bitSKE}
Let $(G,S)$ be a HaarPR and Haar-DDH QGA. 
Then, the following SKE is IND-CPA-secure $1$-bit SKE. 
\begin{itemize}
\item $\KeyGen(1^\secp):$ Generate $g \gets G(1^\secp)$ and output $K = g$.
\item $\Enc(K,b):$ If $b=0$, then generate $|s\rangle \gets \mu$ by using a $1$-design and output a ciphertext $\ct = (|s\rangle, g |s\rangle)$. If $b=1$, then generate $|s\rangle, |s'\rangle \gets \mu$ by using a $1$-design and output $\ct = (|s\rangle, |s'\rangle)$
\item $\Dec(K,\ct):$ Let $\ct = (|\phi\rangle, |\psi\rangle)$. 
Compute $g \otimes I$ on $\ct$ and run the SWAP test between registers. 
Output the result of the SWAP test. 
\end{itemize}
\begin{proof}[Proof of correctness]
If $b = 0$, then $|\psi\rangle = g|s\rangle$, then the registers after applying $g \otimes I$ is 
$(g|s\rangle, g|s\rangle)$. Thus, the SWAP test always outputs $0$. 
On the other hand, if $b=1$, then $|\psi\rangle = |s'\rangle$ is independent of $|\phi\rangle = |s\rangle$. Thus, 
\begin{align}
\Pr[\algo{Dec}(K,\algo{Enc}(K,1)) = 0]
&\leq \mathbb{E}_{g\gets G, |s\rangle,|s'\rangle \gets \mu}\left[\frac{1+|\langle s'|g|s\rangle|^2}{2}\right]
+ \negl(\secp) \label{eq:1-bitSKE:correct1} \\
&\leq \mathbb{E}_{|s\rangle,|s'\rangle \gets \mu}\left[\frac{1+|\langle s'|s\rangle|^2}{2}\right]
+ \negl(\secp) \label{eq:1-bitSKE:correct2} \\
&\leq (1+1/2)/2 + \negl(\secp) \leq 4/5, \label{eq:1-bitSKE:correct3}
\end{align}
where we used 
\cref{lem:design-Haar} for \cref{eq:1-bitSKE:correct1}, 
the fact that, for any $g$, the distribution of $g|s\rangle \gets \mu$ is equivalent to 
 that of $g|s\rangle \gets \mu$ for \cref{eq:1-bitSKE:correct2}, 
and \cref{lem:Haar-orthogonal} for \cref{eq:1-bitSKE:correct3}.
\end{proof}

\begin{proof}[Proof of IND-CPA security]
If the QGA is Haar-PR, then $D'_{\subHaarPR,0} \approx_c D'_{\subHaarPR,1}$ holds in \cref{lem:HaarPR_Q}. 
Notice that $D'_{\subHaarPR,0} = D_{\subHDDH,1}$ in \cref{def:QGA:Haar-DDH}. 
Thus, if the underlying QGA is Haar-PR and Haar-DDH, then $D'_{\subHaarPR,1} \approx_c D_{\subHDDH,0}$, where 
\begin{align*}
D'_{\subHaarPR,1} &:
\text{for } q \in [Q] ~ |s_q\rangle \gets \mu, |s'_q\rangle \gets \mu; 
\text{ return } \{(|s_q\rangle, |s'_q\rangle)^{\otimes t}\}_{q \in [Q]}, \\
D_{\subHDDH,0}     &:
g \gets G, \text{ for } q \in [Q] ~ |s_q\rangle \gets \mu; 
\text{ return } \{(|s_q\rangle, g|s_q\rangle)^{\otimes t}\}_{q \in [Q]}. 
\end{align*}

Let $\cA$ be an adversary against the IND-CPA security
 and let $Q$ be the number of queries $\cA$ making. 
We construct a reduction algorithm $\cB$ distinguishing $D'_{\subHaarPR,1}$ and $D_{\subHDDH,0}$ with $t = 1$ as follows: 
\begin{enumerate}
\item Receive samples $\{(|s_q\rangle, |s'_q\rangle)\}_{q \in [Q]}$ as input, where $|s'_q\rangle$ is $g |s_q\rangle$ with $g \gets G$ or chosen from $\mu$. 
\item Run $\cA$ and simulate the oracle as follows:
\begin{itemize}
\item If the $i$-th query is $0$, then return $\ct_i = (|s_q\rangle, |s'_q\rangle)$.
\item If the $i$-th query is $1$, then generate two independent samples $|\phi\rangle, |\psi\rangle$ by using $1$-design and return $\ct_i = (|\phi\rangle, |\psi\rangle)$. 
\end{itemize}
\item Output $\cA$'s decision. 
\end{enumerate}
If the input samples are chosen from $D_{\subHDDH,0}$, 
 then $\cB$ perfectly simulates the encryption oracle $\Enc(K,\cdot)$, where $K = g \gets G$. 
On the other hand, if the input samples are chosen from $D'_{\subHaarPR,1}$, 
 then $\cB$ statistically simulates the encryption oracle $\Enc(K,1)$. 
Thus, $\cB$'s advantage is statistically close to that of $\cA$ against IND-CPA security. 
This completes the proof. 
\end{proof}

\end{theorem}
Since an IND-CPA-secure multi-bit SKE implies an IND-CPA-secure \emph{quantum-message} SKE~\cite{C:BroJef15} (and the formal proof in \cite[Appendix A]{EPRINT:MorYamYam24}), 
 we have the following corollary. 
\begin{corollary} \label{cor:HaarPR+Haar-DDH->SKE}
Let $(G,S)$ be a HaarPR and Haar-DDH QGA. 
Then, an IND-CPA-secure quantum-message SKE exists. 
\end{corollary}

\section{Discussion on Naor-Reingold-style PRFs from Group Actions} \label{sec:ClassicNRPRF}

Here, we discuss how to weaken algebraic structures of group actions in the existing proofs~\cite{AC:BonKogWoo20,AC:ADMP20}. 
We first briefly review group actions and their notions. 
We then discuss the existing proofs by Boneh~et~al.~\cite{AC:BonKogWoo20} and Alamati~et~al.~\cite{AC:ADMP20}. 

\subsection{Preliminaries}

We first review the definition of group actions. 
\begin{definition}[Group action]
Let $G$ be a group with an identity element $1_G$ and let $S$ be a set. 
Let $\star \colon G \times S \to S$ be a map. 
We say that $(G,S,\star)$ is a group action if the map satisfies the following two properties: 
\begin{enumerate}
\item Identity: For any $s \in S$, we have $1_G \star s = s$. 
\item Compatibility: For any $g,h \in G$ and any $s \in S$, it holds that $(gh) \star s = g \star (h \star s)$. 
\end{enumerate}
% We say that a group $G$ \emph{acts on} a set $S$ if  
\end{definition}

We next review the standard notions of group actions. 
\begin{definition}[Properties of group actions]
\mbox{}
\begin{enumerate}
\item Transitive: $(G,S,\star)$ is said to be \emph{transitive} if for arbitrary $s_1,s_2 \in S$, there exists a group element $g \in G$ satisfying $s_2 = g \star s_1$. 
\item Faithful: $(G,S,\star)$ is said to be \emph{faithful} if for each group element $g \in G$, either $g = 1_G$ or there exists an element $s \in S$ satisfying $s \neq g \star s$. In other words, a group action is \emph{faithful} if $g  = 1_G$ if and only if $s = g \star s$ for all $s \in S$. 
%\item Free: $(G,S,\star)$ is called \emph{free} if for each group element $g \in G$, $g = 1_G$ if and only if there exists some element $s \in S$ satisfying $s = g \star s$. 
\item Free: $(G,S,\star)$ is called \emph{free} if for each group element $g \in G$, if there exists some element $s \in S$ satisfying $s = g \star s$ then $g = 1_G$. Note that if group action is free, then it is also faithful. 
\item Regular: $(G,S,\star)$ is said to be \emph{regular} if it is transitive and free. 
\end{enumerate}
\end{definition}
For an element $s \in S$, we consider a mapping $f_s \colon g \in G \mapsto g \star s \in S$. 
We also consider, for an element $g \in G$, a mapping $L_g \colon s \in S \mapsto g \star s \in S$. 
We note that, for any $s \in S$, if a group action is \emph{transitive} (or \emph{free}, resp.),  then $f_s$ is subjective (or injective, resp.). 
We also note that, if a group action is \emph{faithful}, then for any $g \neq h \in G$, $L_g \neq L_h$.
\begin{lemma} \label{lem:GA:PerfectRandom}
Suppose that $G$ is finite and a group action $(G,S,\star)$ is \emph{transitive and faithful}. 
Then, 
 for any $s_0 \in S$, 
 then the distribution of $s_i \gets S$ is equivalent to that of $g_i \star s_0$ with $g_i \gets G$. 
\end{lemma}
\begin{proof}
The proof is easily obtained by considering a subgroup $H = \{g : g \star s_0 = s_0\}$ and left cosets $\{g H\}$ induced by $H$ and using the facts in above. 
\end{proof}
% If a group action is \emph{transitive and free}, 
%  for any $s \in S$, $f_s \colon g \mapsto g \star s_0$ is a bijection and we have $|G| = |S|$. 
% if a group action is \emph{transitive}, then $f_s$ is subjective for any $s$. 
% It is also easy to see that if a group action is \emph{free}, then $f_s$ is injective for any $s$. 

We then review \emph{effective group actions} in \cite{AC:ADMP20}.  
\begin{definition}[{Effective group actions (EGAs) \cite[Definition 3.4]{AC:ADMP20}}]
We say that a group action $(G,S,\star)$ is \emph{effective} if the following properties are satisfied: 
\begin{enumerate}
\item The group $G$ is finite and there exist efficient algorithms for:
    \begin{enumerate}
    \item Membership testing, that is, to decide if a given bit-string represents a valid group element in $G$ or not. 
    \item Equality testing, that is, to decide if two bit-strings represent the same group element in $G$ or not. 
    \item Sampling, that is, to sample an element $g$ from a distribution that is statistically close to the uniform over $G$. 
    \item Operation, that is, to compute $gh$ from $g,h \in G$. 
    \item Inversion, that is, to compute $g^{-1}$ from $g \in G$.
    \end{enumerate}
\item The set $X$ is finite and there exist efficient algorithms for:
    \begin{enumerate}
    \item Membership testing, that is, to decide if a given bit-string represents a valid set element in $S$ or not. 
    \item Unique representation, that is, given any $s \in S$, to compute a string $\hat{s}$ that canonically represents $s$.
    \end{enumerate}
\item Origin: There exists an element $s_0 \in S$, called the origin, such that its representation is known in public. 
\item Operation $\star$: There exists an efficient algorithm that takes $g \in G$ and $s \in S$ and outputs $g \star s$. 
\end{enumerate}
\end{definition}

We next define several computational assumptions of group actions. 
\begin{definition}[Assumptions]
\mbox{}
\begin{enumerate}
\item Pseudorandom: A GA is called \emph{pseudorandom (PR)} if 
\[
\{(s_0, g \star s_0): g \gets G\} \approx_c \{(s_0, s_1): s_1 \gets S\}. 
\]
\item Weakly pseudorandom: A GA is called \emph{weakly pseudorandom (wPR)} if for any polynomial $Q = Q(\secp)$, 
\[
\{(s_i, g \star s_i): g \gets G, s_i \gets S\} \approx_c \{(s_i, s'_i): s_i,s'_i \gets S\}. 
\]
\item Decisional Diffie-Hellman (DDH): A GA is called \emph{Decisional Diffie-Hellman (DDH)} if 
\[
\{(s_0, \tilde{g} \star s_0, g \star s_0, (\tilde{g}g) \star s_0): \tilde{g},g \gets G\}
\approx_c \{(s_0, \tilde{g} \star s_0, g \star s_0, h \star s_0): \tilde{g},g,h \gets G\}.
\]
\item Naor-Reingold (NR): A GA is called \emph{Naor-Reingold (NR)} if for any polynomial $Q = Q(\secp)$, 
\[
\{(g_i \star s_0, (\tilde{g} g_i) \star s_0): \tilde{g},g_i \gets G \} 
\approx_c \{(g_i \star s_0, h_i \star s_0): g_i,h_i \gets G\}. 
\]
\end{enumerate}

\end{definition}

\paragraph{NR-style PRF.}
Let $(G,S,\star)$ be an EGA. 
We define $f \colon G^{\ell+1} \times \bit^\ell \to S$ as 
\[
f_{g_0,\dots,g_\ell}(x_1,\dots,x_\ell) \coloneqq
 (g_\ell^{x_\ell} \cdot \dots \cdot g_1^{x_1} \cdot g_0) \star s_0. 
\]
We say that this function is PRF is this $f$ is computationally indistinguishable with a random function $f' \colon x \in \bit^\ell \mapsto s_x \in S$, where $s_x \gets S$ for each $x \in \bit^\ell$. 
%We note that there are two ways to define random functions $\bit^\ell \to S$. 
%One is $f_1 \colon x \in \bit^\ell \mapsto s_x \in S$, where $s_x \gets S$ for each $x \in \bit^\ell$. 
%The other is $f_2 \colon x \in \bit^\ell \mapsto g_x \star s_0 \in S$, where $g_x \gets G$ for each $x \in \bit^\ell$. 

By adopting the proof for the NR-style PRFSG~(\cref{thm:PRNR->PRFSG}), 
 we obtain the following theorem for the NR-style PRF $f$: 
\begin{theorem} \label{thm:PR+NR->PRF}
Let $(G,S,\star)$ be an EGA. 
If it is PR and NR, then the function $f$ is a PRF. 
\end{theorem}
We then review the proofs in \cite{AC:BonKogWoo20} and \cite{AC:ADMP20}
 and weaken the requirements of them. 

\subsection{BKW20 Proof}
Boneh~et~al.~\cite{AC:BonKogWoo20} assumed that a group action is transitive and faithful and $G$ is commutative. 
\begin{theorem}[{\cite[Section 8]{AC:BonKogWoo20}, adapted}]
Let $(G,S,\star)$ be an EGA. 
Suppose that the EGA is transitive and faithful and $G$ is commutative. 
If the EGA is DDH, then the function $f$ is a PRF. 
\end{theorem}

The following lemma (\cref{lem:EGA:com+DDH->NR}) shows that if $G$ is commutative and the EGA is DDH, then it is NR. 
Combining the lemma with \cref{thm:PR+NR->PRF}, we obtain the following corollary. 
\begin{corollary}
Let $(G,S,\star)$ be an EGA. 
If $G$ is commutative, and the EGA is PR and DDH, then the function $f$ is a PRF. 
\end{corollary}
% Instead of this, if we assume pseudorandomness of group action, then we do not require the group action to be transitive and faithful. 
% \begin{theorem}
% If $(G,S,\star)$ is an EGA, $G$ is commutative, and the PR and DDH assumptions hold, then the function $f$ is a PRF. 
% \end{theorem}
% We first show that .... 
\begin{lemma} \label{lem:EGA:com+DDH->NR}
Let $(G,S,\star)$ be an EGA. 
If $G$ is commutative and the EGA is DDH, then the EGA is NR. 
\end{lemma}
\begin{proof}
Let us consider hybrid distributions $\bar{D}_i$: 
For $j = 1,\dots,i$, $(a_i,b_i) = (g_i \star s_0, h_i \star s_0)$
 and $j = i+1,\dots,Q$,  $(a_i,b_i) = (g_i \star s_0, \tilde{g}g_i \star s_0)$. 
By using the following claim, we have 
$\bar{D}_0 \approx_c \bar{D}_1 \approx_c \dots \approx_c \bar{D}_Q$ 
if the DDH assumption holds. 
\end{proof}
\begin{claim}
If $(G,S,\star)$ is an EGA, $G$ is commutative, and the DDH assumptions hold,
for $i = 1,\dots,Q$, 
$\bar{D}_{i-1} \approx_c \bar{D}_i$ holds. 
\end{claim}
\begin{proof}
Suppose that there exists $\cA$ distinguishing $\bar{D}_{i-1}$ from $\bar{D}_i$. 
We construct an adversary $\cB$ against the DDH assumption as follows:
\begin{itemize}
\item Given a sample $(s_0, \tilde{g} \star s_0, g \star s_0, h \star s_0)$, where $h = \tilde{g}g$ or random, $\cB$ prepares a sample $\{(a_j,b_j)\}_{j \in [Q]}$ as follows:  
\begin{itemize}
\item for $j=1,\dots,i-1$, take random $g_j,h_j \gets G$ and set $(a_j,b_j) \coloneqq (g_j \star s_0, h_j \star s_0)$; 
\item for $j=i$, set $(a_j,b_j) = (g \star s_0, h \star s_0)$;  
\item for $j=i+1,\dots,Q$, take random $g_j \gets G$ and set $(a_j,b_j) = (g_j \star s_0, g_j \tilde{g} \star s_0)$. 
\end{itemize}
\item It runs $\cA$ on input $\{(a_j,b_j)\}_{j \in [Q]}$ and outputs $\cA$'s decision. 
\end{itemize}
We note that, due to commutativity of $G$, the last $Q-i$ samples are equivalent to $(g_j \star s_0, \tilde{g} g_j \star s_0)$. 
If $h = \tilde{g}g$, then $(a_j,b_j) = (g \star s_0, (\tilde{g} g) \star s_0)$ and $\cB$ perfectly simulates the distribution $\bar{D}_{i-1}$ since $G$ is a group and the distribution of 
If $h$ is random, then $\cB$ perfectly simulates the distribution $\bar{D}_{i}$. 
Thus, $\cB$'s advantage is equivalent to $\cA$'s advantage distinguishing $\bar{D}_{i-1}$ and $\bar{D}_i$. 
\end{proof}

\subsection{ADMP20 Proof}
Alamati~et~al.~\cite{AC:ADMP20} assumed that a group action is weakly pseudorandom and $G$ is regular and commutative. 
\begin{theorem}[{\cite[Section 3.1 and Section 4.4]{AC:ADMP20}, adapted}]
Let $(G,S,\star)$ be an EGA. 
Suppose that the EGA is regular and $G$ is commutative.\footnote{\cite[Section 3.1 and Section 4.4]{AC:ADMP20}}
If the EGA is weakly pseudorandom, then the function $f$ is a PRF. 
\end{theorem}

As in our discussion in the introduction, we do not require commutativity of $G$. 
It is easy to show that if the EGA is wPR and PR, then the EGA is NR (\cref{lem:EGA:wPR+PR->NR} below). 
In addition, due to \cref{lem:GA:PerfectRandom}, if the EGA is wPR and $G$ is transitive and faithful, then the EGA is NR (\cref{lem:EGA:trans+faith+wPR->NR}). 
Thus, we obtain the following corollary of \cref{thm:PR+NR->PRF}. 
\begin{corollary}
Let $(G,S,\star)$ be an EGA.
\begin{itemize}
\item If it is PR and wPR, then the function $f$ is a PRF. 
\item If it is wPR and $G$ is transitive and faithful, then the function $f$ is a PRF. 
\end{itemize}
\end{corollary}

\begin{lemma} \label{lem:EGA:wPR+PR->NR}
Let $(G,S,\star)$ be an EGA.
If it is wPR and PR, then it is NR. 
\end{lemma}
\begin{proof}
It is easy to see that if the EGA is PR, then we can replace ``$s_i \gets S$'' with ``$g_i \star s_0$ with $g_i \gets G$''. Thus, we have that 
\begin{align}
 &\{(s_i,g \star s_i): g \gets G, s_i \gets S\}_{i \in [Q]}
 \approx_c 
 \{(g_i \star s_0 ,g g_i \star s_0): g, g_i \gets G\}_{i \in [Q]}  \label{cind:wPR+PR->NR:1} \\ 
 &\{(s_i,h_i \star s_i): h_i \gets G, s_i \gets S\}_{i \in [Q]}
 \approx_c
 \{(g_i \star s_0, h_i g_i \star s_0): h_i, g_i \gets G\}_{i \in [Q]}. \label{cind:wPR+PR->NR:2} 
\end{align}
%Since $G$ is a group, 
%$ \{(g_i \star s_0, h_i g_i \star s_0): h_i, g_i \gets G\}_{i \in [Q]}$ 
%is equivalent to $ \{(g_i \star s_0, h_i \star s_0): h_i, g_i \gets G\}_{i \in [Q]}$. 
We obtain that 
\begin{align*}
&\{(g_i \star s_0 ,g g_i \star s_0): g, g_i \gets G\}_{i \in [Q]} \\
&\approx_c \{(s_i,g \star s_i): g \gets G, s_i \gets S\}_{i \in [Q]} &\text{(from \cref{cind:wPR+PR->NR:1})} \\ 
&\approx_c \{(s_i,s'_i): s_i,s'_i \gets S\}_{i \in [Q]}              &\text{(from wPR)} \\ 
&\approx_c \{(s_i, h_i \star s_i): h_i \gets G, s_i \gets S\}_{i \in [Q]} &\text{(from wPR)} \\
&\approx_c \{(g_i \star s_0, h_i g_i \star s_0): h_i, g_i \gets G\}_{i \in [Q]} &\text{(from \cref{cind:wPR+PR->NR:2})} \\
&\equiv \{(g_i \star s_0, h_i \star s_0): h_i, g_i \gets G\}_{i \in [Q]} &\text{($G$ is a group)}, 
\end{align*}
where we apply wPR $Q$-times to obtain third computational indistinguishability. 
\end{proof}
Recall that if $G$ is transitive and faithful, then an EGA $(G,S,\star)$ is perfectly PR (\cref{lem:GA:PerfectRandom}). 
Thus, we obtain the following corollary. 
\begin{corollary} \label{lem:EGA:trans+faith+wPR->NR}
Let $(G,S,\star)$ be an EGA.
If $G$ is transitive and faithful, and the EGA is wPR, then the EGA is NR.    
\end{corollary}

% \begin{lemma} \label{lem:EGA:trans+faith+wPR->NR}
% \end{lemma}
% \begin{proof}
% Since $G$ is transitive and faithful, \cref{lem:GA:PerfectRandom} implies that 
% \begin{align*}
% \{(s_i,s'_i): s_i,s'_i \gets S\}_{i \in [Q]}
%  \equiv \{(s_i,h_i \star s_i): h_i \gets G, s_i \gets S\}_{i \in [Q]}. 
% \end{align*}
% By replacing ``$s_i \gets S$'' with ``$g_i \star s_0$ with $g_i \gets G$'' via \cref{lem:GA:PerfectRandom}, we have that 
% \begin{align*}
%  &\{(s_i,g \star s_i): g \gets G, s_i \gets S\}_{i \in [Q]}
%  \equiv 
%  \{(g_i \star s_0 ,g g_i \star s_0): g, g_i \gets G\}_{i \in [Q]} \\ 
%  &\{(s_i,h_i \star s_i): h_i \gets G, s_i \gets S\}_{i \in [Q]}
%  \equiv 
%  \{(g_i \star s_0, h_i g_i \star s_0): h_i, g_i \gets G\}_{i \in [Q]}. 
% \end{align*}
% Since $G$ is a group, 
% $ \{(g_i \star s_0, h_i g_i \star s_0): h_i, g_i \gets G\}_{i \in [Q]}$ 
% is equivalent to $ \{(g_i \star s_0, h_i \star s_0): h_i, g_i \gets G\}_{i \in [Q]}$. 
% Therefore, we obtain that 
% \begin{align*}
% \{(g_i \star s_0 ,g g_i \star s_0): g, g_i \gets G\}_{i \in [Q]}
% &\equiv \{(s_i,g \star s_i): g \gets G, s_i \gets S\}_{i \in [Q]} \\ 
% &\approx_c \{(s_i,s'_i): s_i,s'_i \gets S\}_{i \in [Q]} \\ 
% &\equiv \{(s_i, h_i \star s_i): h_i \gets G, s_i \gets S\}_{i \in [Q]} \\
% &\equiv \{(g_i \star s_0, h_i g_i \star s_0): h_i, g_i \gets G\}_{i \in [Q]} \\
% &\equiv \{(g_i \star s_0, h_i \star s_0): h_i, g_i \gets G\}_{i \in [Q]}
% \end{align*}
% as we wanted. 
% \end{proof}

\end{document}